\documentclass[11pt]{article}

\usepackage[T1]{fontenc}
\usepackage[margin=1in]{geometry}
\usepackage{amsmath}
\usepackage{amssymb}
\usepackage{amsfonts}
\usepackage{amsthm}
\usepackage{amsopn}
\usepackage{graphicx}
\usepackage{xcolor}
\usepackage{xspace}
\usepackage{ifpdf}
\usepackage[hidelinks]{hyperref}
\usepackage[capitalize,nameinlink]{cleveref}
\numberwithin{equation}{section}
\crefformat{equation}{\textup{#2(#1)#3}}
\crefrangeformat{equation}{\textup{#3(#1)#4--#5(#2)#6}}
\crefmultiformat{equation}{\textup{#2(#1)#3}}{ and \textup{#2(#1)#3}}{, \textup{#2(#1)#3}}{, and \textup{#2(#1)#3}}
\crefrangemultiformat{equation}{\textup{#3(#1)#4--#5(#2)#6}}{ and \textup{#3(#1)#4--#5(#2)#6}}{, \textup{#3(#1)#4--#5(#2)#6}}{, and \textup{#3(#1)#4--#5(#2)#6}}
\Crefformat{equation}{#2Equation~\textup{(#1)}#3}
\Crefrangeformat{equation}{Equations~\textup{#3(#1)#4--#5(#2)#6}}
\Crefmultiformat{equation}{Equations~\textup{#2(#1)#3}}{ and \textup{#2(#1)#3}}{, \textup{#2(#1)#3}}{, and \textup{#2(#1)#3}}
\Crefrangemultiformat{equation}{Equations~\textup{#3(#1)#4--#5(#2)#6}}{ and \textup{#3(#1)#4--#5(#2)#6}}{, \textup{#3(#1)#4--#5(#2)#6}}{, and \textup{#3(#1)#4--#5(#2)#6}}

\ifpdf
  \DeclareGraphicsExtensions{.pdf,.png,.jpg,.eps}
\else
  \DeclareGraphicsExtensions{.eps}
\fi

\theoremstyle{plain}
\newtheorem{theorem}{Theorem}[section]
\newtheorem{lemma}[theorem]{Lemma}
\newtheorem{corollary}[theorem]{Corollary}
\newtheorem{proposition}[theorem]{Proposition}
\newtheorem{claim}[theorem]{Claim}
\theoremstyle{definition}
\newtheorem{definition}[theorem]{Definition}
\theoremstyle{remark}
\newtheorem{remark}[theorem]{Remark}

\crefname{hypothesis}{Hypothesis}{Hypotheses}

\newenvironment{keywords}{\par\smallskip\noindent\textbf{Key words.}\ }{\par\smallskip}
\newenvironment{MSCcodes}{\par\smallskip\noindent\textbf{MSC codes.}\ }{\par\smallskip}
\newcommand{\funding}[1]{\par Funding: #1}
\newcommand{\email}[1]{\href{mailto:#1}{#1}}

\title{The Discrepancy of Shortest Paths\thanks{ A preliminary version appeared in the 51st International Colloquium on Automata, Languages, and Programming (ICALP 2024). \url{https://doi.org/10.4230/LIPIcs.ICALP.2024.27}
\funding{This work has been partially supported by NSF through AF 2153680, CCF-2118953, CCF-2208663, IIS-2229876, DMS-2220271, DMS-2311064, CRCNS-2207440, and CNS 2433628, and through Google Seed Fund grant, Google Research Scholar Award, and Dean Research Seed Fund.}}}

\author{
Greg Bodwin
\thanks{University of Michigan
  (\email{bodwin@umich.edu}).}
\quad
Chengyuan Deng
\thanks{Rutgers University
  (\email{cd751@rutgers.edu}).}
\quad
Jie Gao
\thanks{Rutgers University
  (\email{jg1555@rutgers.edu}).}
\\[0.5em]
Gary Hoppenworth
\thanks{University of Michigan
  (\email{garytho@umich.edu})}
\quad
Jalaj Upadhyay
\thanks{Rutgers University
  (\email{jalaj.upadhyay@rutgers.edu}).}
\quad
Chen Wang
\thanks{Rensselaer Polytechnic Institute
  (\email{wangc33@rpi.edu}).}
}

\date{}

\usepackage{multirow}

\usepackage{thmtools}
\usepackage{thm-restate}
\newtheorem{observation}[theorem]{Observation}

\def\R{\mathcal{R}}
\def\S{\mathcal{S}}
\DeclareMathOperator{\disc}{disc}
\DeclareMathOperator{\tr}{tr}
\DeclareMathOperator{\herdisc}{herdisc}
\def\discv{\textnormal{disc}_v}
\def\disce{\textnormal{disc}_e}
\def\hdiscv{\textnormal{herdisc}_v}
\def\hdisc{\textnormal{herdisc}}
\def\hdisce{\textnormal{herdisc}_e}

\newcommand{\card}[1]{\left| #1 \right|}

\def\eps{\varepsilon}

\newcommand{\paren}[1]{\ensuremath{\left(#1\right)}}
\newcommand{\bracket}[1]{\left[#1\right]}
\newcommand{\expect}[1]{\mathbb{E}\bracket{#1}}

\newcommand{\norm}[1]{{\left\|#1\right\|}}
\newcommand{\rr}{\mathbb{R}}

\makeatletter
\newcommand{\thickhline}{%
    \noalign {\ifnum 0=`}\fi \hrule height 1.3pt
    \futurelet \reserved@a \@xhline
}
\makeatother

\def\reals{\mathbb{R}}
\newcommand{\prob}[1]{\Pr\left[#1\right]}

\newcommand{\Oish}{\widetilde{O}}
\newcommand{\Omegaish}{\widetilde{\Omega}}
\newcommand{\Thetaish}{\widetilde{\Theta}}
\providecommand{\Vec}[1]{\vec{#1}}

\ifpdf
\hypersetup{
  pdftitle={The Discrepancy of Shortest Paths},
  pdfauthor={G. Bodwin, C. Deng, J. Gao, G. Hoppenworth, J. Upadhyay, and C. Wang}
}
\fi

\begin{document}

\maketitle

\begin{abstract} 
The \textit{hereditary discrepancy} of a set system is a quantitative measure of the pseudorandom properties of the system.
Roughly speaking, hereditary discrepancy measures how well one can $2$-color the elements of the system so that each set contains approximately the same number of elements of each color.
Hereditary discrepancy has numerous applications in computational geometry, communication complexity and derandomization. More recently, the hereditary discrepancy of the set system of \textit{shortest paths} has found applications in differential privacy [Chen et al.~SODA 23].

The contribution of this paper is to improve the upper and lower bounds on the hereditary discrepancy of set systems of unique shortest paths in graphs. 
In particular, we show that any system of unique shortest paths in an undirected weighted graph has hereditary discrepancy $O(n^{1/4})$, and we construct lower bound examples demonstrating that this bound is tight up to $\text{polylog } n$ factors.
Our lower bounds hold even for planar graphs and bipartite graphs and improve a previous lower bound of $\Omega(n^{1/6})$ obtained by applying the trace bound of Chazelle and Lvov [SoCG'00] to a classical point-line system of  Erd\H{o}s.

As applications, we improve the lower bound on the additive error for differentially-private all pairs shortest distances from $\Omega(n^{1/6})$ [Chen et al.~SODA 23] to $\widetilde{\Omega}(n^{1/4})$, and we improve the lower bound on additive error for the differentially-private all sets range queries problem to $\widetilde{\Omega}(n^{1/4})$, which is tight up to $\text{polylog } n$ factors [Deng et al. WADS 23].
\end{abstract}



\setcounter{page}{1}

\section{Introduction}

In graph algorithms, a fundamental problem is to efficiently compute the distance or shortest path information of a given input graph.
Over the last decade or so, the community has increasingly sought a principled understanding of the \textit{combinatorial structure} of shortest paths, with the goal to exploit this structure in algorithm design.
That is, in various graph settings, we can ask:

\begin{quote}
\textit{What notable structural properties hold for \textbf{shortest} path systems, that do not necessarily hold for \textbf{arbitrary} path systems?}
\end{quote}

\noindent The following are a few of the major successes of this line of work:

\begin{itemize}
\item An extremely popular strategy in the literature is to use \textit{hitting sets}, in which we (often randomly) generate a set of nodes $S$ and argue that it will hit the shortest path for every pair of nodes that are sufficiently far apart.
Hitting sets rarely exploit any structure of shortest paths, as evidenced by the fact that most hitting set algorithms generalize immediately to arbitrary set systems.
However, they have inspired a successful line of work into graphs of bounded \textit{highway dimension} \cite{abraham2010highway, bast2016route, blum2022sparse}; very roughly, these are graphs whose shortest paths admit unusually efficient hitting sets of a certain kind.

\item Shortest paths exhibit the notable structural property of \textit{consistency}, i.e., any subpath of the shortest path is itself the shortest path.
This fact is used throughout the literature on graph algorithms \cite{Bodwin19,CE06, cormen2022introduction}, including, e.g., \ in the classic Floyd-Warshall algorithm for All-Pairs Shortest Paths.
A recent line of work has sought to characterize the additional structure exhibited by shortest path systems beyond consistency \cite{AW23, amiri2020disjoint, Bodwin19, chudnovsky2023structure, cizma2022geodesic, cizma2023irreducible, CE06}.

\item Planar graphs have received special attention within this research program, and planar shortest path systems carry some notable additional structure.
For example, it is known that planar shortest paths have unusually efficient tree coverings \cite{balzotti2022non, chang2023covering}, and that their shortest paths can be compressed into surprisingly small space \cite{CGMW18, chang2022near}. 
Shortest path algorithms also often benefit from more general structural facts about planar graphs, such as separator theorems \cite{har2011simple, henzinger1997faster}.
\end{itemize}

The main result of this paper is a new structural separation between shortest path systems and arbitrary path systems, expressed through the lens of \textit{discrepancy theory}.
We will come to formal definitions of discrepancy in just a moment, but at a high level, discrepancy has been described as a quantitative measure of the combinatorial pseudorandomness of a discrete system~\cite{chung1989quasi}, and it has widespread applications in discrete and computational geometry, random sampling and derandomization, communication complexity, and much more\footnote{We refer to the excellent textbooks of Alexander, Beck, and Chen~\cite{alexander200413}, Chazelle \cite{Chazelle2000-rd}, and Matou{\v s}ek~\cite{Matousek1999-fz} for discussion and further applications. }.
We will show the following:
\begin{theorem} [Main Result, Informal]
The discrepancy of unique shortest path systems in weighted graphs is inherently smaller than the discrepancy of arbitrary path systems in graphs.
\end{theorem}

This separation between unique shortest paths and arbitrary paths is due to the structural property of  \textit{consistency} of unique shortest path systems, which is well-studied in the literature \cite{Bodwin19, CE06, cormen2022introduction}.

Our results can be placed within a larger context of prior work in computational geometry. A classical topic in this area is to determine the discrepancy of incidence structures between points and geometric range spaces such as axis-parallel rectangles, half-spaces, lines, and curves  (cf.\ \cite[Section 1.5]{Chazelle2000-rd}). These results have been used to show lower bounds for geometric range searching~\cite{muthukrishnan2012optimal,Toth2017-du}. 

Indeed, systems of unique shortest paths in graphs capture some of the geometric range spaces studied in prior work. For instance, arrangements of straight lines in Euclidean space can be interpreted as systems of unique shortest paths in an associated graph, implying a relation between the discrepancies of these two set systems. This connection has recently found applications in the study of differential privacy on shortest path distance and range query algorithms~\cite{ghazi2022differentially,deng2023differentially}. 

More generally, discrepancy on graphs has also found applications in proving tight lower bounds on answering cut queries on graphs~\cite{eliavs2020differentially,liu2023optimal}. 

\subsection{Formal Definitions of Discrepancy}
We first define some notation that we use throughout the paper. We use the letter $\rr$ to denote the set of real numbers and $\mathbb N$ to denote the set of natural numbers. For $n \in \mathbb N$, we use the notation $[n]$ to denote the set $\{1, \cdots, n\}$. For a real number $r \in \rr$, we use $|r|$ to denote the absolute value.

For a vector $v$, we use the notation $v_i$ to denote its $i$-th coordinate, and for a matrix $A$, we use the notation $A_{i,j}$ to denote its $(i,j)$-th entry. For a vector $v \in \rr^n$ and any positive number $p \in \mathbb N_{\geq 0}$, we use different vector norms:
\begin{align*}
\|v \|_p := 
	\left( |v_1|^p + \cdots + |v_n|^p \right)^{1/p} 
\end{align*}
where $|\cdot|$ denote the absolute value. By continuity, $\|v\|_\infty:=\max_{1 \leq i \leq n} |v_i|$. For positive $p,q \in \mathbb N_{\geq 0}$ and any matrix $A \in \rr^{m \times n}$, 
\[\|{A}\|_{p \to q} = \max \frac{\| Ax\|_q }{\|x\|_p}.\]

We reserve the symbol $G$ for graphs with vertex set $V$ and edge set $E$. We reserve the symbol $\Pi$ to denote a given path system. For a subset $U \subseteq V$, we use the notation $G[U]$ to denote the subgraph induced by the vertex set $U$.

Throughout the paper, we use the $\widetilde O$ and $\widetilde \Omega$ to hide poly-logarithmic factors in the input parameter, $n$.

We first collect the basic definitions needed to understand this paper.

\begin{definition} [Edge and Vertex Incidence Matrices]
\label{def:disc-herdisc}
Given a graph $G = (V, E)$ and a set of paths $\Pi$ in $G$, the associated {\em vertex incidence matrix} is given by $A \in \{0,1\}^{|\Pi| \times |V|}$, where for each $v \in V$ and $\pi \in \Pi$ the corresponding entry is
\[A_{\pi, v} = \begin{cases}
1 & \text{if } v \in \pi\\
0 & \text{if } v \notin \pi.
\end{cases}
\]
The associated {\em edge incidence matrix} is given by a binary matrix $A \in \{0,1\}^{|\Pi| \times |E|}$, where for each $e \in E$ and $\pi \in \Pi$ the corresponding entry is
\[A_{\pi, e} = \begin{cases}
1 & \text{if } e \in \pi\\
0 & \text{if } e \notin \pi.
\end{cases}
\]
\end{definition}

\begin{definition} [Discrepancy and Hereditary Discrepancy]
Given a matrix $A \in \rr^{m \times n}$, its \emph{discrepancy} is the quantity
\[\disc(A) = \min_{x \in \{1, -1\}^n} \|Ax\|_\infty.\]
Its \emph{hereditary discrepancy} is the maximum discrepancy of any submatrix $A_Y$ obtained by keeping all rows but only a subset $Y \subseteq [n]$ of the columns; that is,
\[\hdisc(A) = \max_{Y \subseteq [n]} \disc(A_Y) \]
\end{definition}

For a system of paths $\Pi$ in a graph $G$, we will write $\discv(\Pi)$ and $\hdiscv(\Pi)$ to denote the discrepancy and the hereditary discrepancy of its vertex incidence matrix, and $\disce(\Pi)$ and $\hdisce(\Pi)$ to denote the discrepancy and the hereditary discrepancy of its edge incidence matrix.

For intuition, the vertex discrepancy of a system of paths $\Pi$ can be equivalently understood as follows.
Suppose that we color each node in $G$ either red or blue, with the goal to balance the red and blue nodes on each path as evenly as possible.
The discrepancy associated to that particular coloring is the quantity
\[\max_{\pi \in \Pi} \bigg| \ \left|\left\{v \in \pi \ \mid \ v \text{ is colored red}\right\}\right| - \left|\left\{v \in \pi \ \mid \ v \text{ is colored blue}\right\}\right| \ \bigg|.\]

The discrepancy of the system $\Pi$ is the minimum possible discrepancy over all colorings.
The hereditary discrepancy is the maximum discrepancy taken over all \textit{induced path subsystems} $\Pi'$ of $\Pi$; that is, $\Pi'$ is obtained from $\Pi$ by focusing only on the colors\footnote{In the coloring interpretation, hereditary discrepancy allows a different choice of coloring for each subsystem $\Pi'$, rather than fixing a coloring for $\Pi$ and considering the induced coloring on each $\Pi'$.} of a subset of vertices $Y\subseteq V$. Edge discrepancy can be understood in a similar way, coloring edges rather than vertices.

Another related quantity is the $\ell_2$-discrepancy and $\ell_2$-hereditary discrepancy, which use $\ell_2$ norm instead of $\ell_{\infty}$ norm in the definition. 
\begin{definition} [$\ell_2$-Discrepancy and $\ell_2$-Hereditary Discrepancy]
Given a matrix $A \in \rr^{m \times n}$, its \emph{$\ell_2$-discrepancy} is the quantity:
\[\disc_2(A) = \min_{x \in \{1, -1\}^n} \frac{1}{\sqrt{m}}\|Ax\|_2.\] 
Its \emph{$\ell_2$-hereditary discrepancy} is the maximum $\ell_2$-discrepancy of any sub-matrix $A_Y$ obtained by keeping all rows but only a subset $Y \subseteq [n]$ of the columns; that is,
\[\hdisc_2(A) = \max_{Y \subseteq [n]} \discv(A_Y) \]
\end{definition}
The $\ell_2$-discrepancy takes a normalization term $\sqrt{m}$ such that $\disc(A)\geq \disc_2(A)$.

For a system of paths $\Pi$ in a graph $G$, we will write $\disc_{v,2}(\Pi)$ and $\hdisc_{v,2}(\Pi)$ to denote the $\ell_2$-discrepancy and the $\ell_2$-hereditary discrepancy of its vertex incidence matrix, and $\disc_{e,2}(\Pi)$ and $\hdisc_{e,2}(\Pi)$ to denote the $\ell_2$-discrepancy and the $\ell_2$-hereditary discrepancy of its edge incidence matrix.

\subsection{Our Results}

Our main result is an upper and lower bound on the hereditary discrepancy of unique shortest path systems in weighted graphs, which match up to hidden $\text{polylog } n$ factors.
\begin{theorem} [Main Result] \label{thm:intromain} ~
\begin{itemize}
	\item \textbf{ \textup{(Upper Bound)}} For any $n$-node undirected weighted graph $G$ with a unique shortest path between each pair of nodes, there exists a polynomial-time algorithm that finds a coloring for the system of shortest paths $\Pi$ such that:
\[\hdiscv(\Pi) \le \Oish(n^{1/4}) \qquad \text{and} \qquad \hdisce(\Pi) \le \Oish(n^{1/4}).\]

	\item \textbf{\textup{ (Lower Bound)}} There are examples of $n$-node undirected weighted graphs $G$ with a unique shortest path between each pair of nodes in which this system of shortest paths $\Pi$ has $\hdiscv(\Pi) \ge \Omegaish(n^{1/4})$ \text{and}  $\hdisce(\Pi) \ge \Omegaish(n^{1/4}).$
In fact, in these lower bound examples we can take $G$ to be planar or bipartite.
\end{itemize}
\end{theorem}

\begin{table}[ht!]
\renewcommand{\arraystretch}{1.5}
\caption{ Overview of vertex/edge (hereditary) discrepancy on general undirected graphs and special families of graphs: tree, bipartite, and planar graphs. Here $n$ is the number of vertices of the graph.
}
\centering
{
\vspace{0.5em}
\label{tab:res-summary-undirected}
\begin{tabular}{c|c|c|c|c|c} 
\thickhline
 &  & Tree & Bipartite & Planar & Undirected Graph   \\ 
\thickhline

\multirow{2}{*}{\begin{tabular}[c]{@{}c@{}}Vertex \end{tabular}}   & \multirow{1}{*}{\begin{tabular}[c]{@{}c@{}}disc\end{tabular}} & \multirow{1}{*}{\begin{tabular}[c]{@{}c@{}}$\Theta(1)$ \end{tabular}} & \multirow{1}{*}{\begin{tabular}[c]{@{}c@{}}$\Theta(1)$ \end{tabular}} & $ O(n^{1/4})$ &   $  \widetilde{\Theta}(n^{1/4})$   \\

& \multirow{1}{*}{\begin{tabular}[c]{@{}c@{}}herdisc \end{tabular}}  & \multirow{1}{*}{\begin{tabular}[c]{@{}c@{}}$\Theta(1)$ \end{tabular}} & \multirow{1}{*}{\begin{tabular}[c]{@{}c@{}}$\widetilde{\Theta}(n^{1/4})$ \end{tabular}}& $\widetilde{\Theta}(n^{1/4})$  & $\Omega(n^{1/6})\cite{Chazelle2000-ld} \rightarrow \widetilde{\Theta}(n^{1/4}) $  \\

\hline

\multirow{2}{*}{\begin{tabular}[c]{@{}c@{}}Edge \end{tabular}}   & \multirow{1}{*}{\begin{tabular}[c]{@{}c@{}}disc\end{tabular}} & \multirow{1}{*}{\begin{tabular}[c]{@{}c@{}}$\Theta(1)$ \end{tabular}} & \multirow{1}{*}{\begin{tabular}[c]{@{}c@{}}$\Theta(1)$ \end{tabular}} & $ O(n^{1/4})$ &  $O(n^{1/4})$    \\

& \multirow{1}{*}{\begin{tabular}[c]{@{}c@{}}herdisc \end{tabular}}  & \multirow{1}{*}{\begin{tabular}[c]{@{}c@{}}$\Theta(1)$ \end{tabular}} & \multirow{1}{*}{\begin{tabular}[c]{@{}c@{}}$\widetilde{\Theta}(n^{1/4})$ \end{tabular}}  & $\widetilde{\Theta}(n^{1/4})$  &  $\Omega(n^{1/6})\cite{Chazelle2000-ld} \rightarrow \widetilde{\Theta}(n^{1/4})$   \\

\thickhline

\end{tabular}
}
\end{table}

\begin{table}[ht!]
\renewcommand{\arraystretch}{1.5}
\caption{ Overview of vertex/edge (hereditary) discrepancy on general directed graphs and special families of graph: tree, bipartite and planar graphs. Here $n$ is the number of vertices of the graph and $m$ is the number of edges. $D$ is the graph diameter or the longest number of hops of paths considered.}
\centering
{
\vspace{0.5em}
\label{tab:res-summary}
\begin{tabular}{c|c|c|c|c|c} 
\thickhline
 &  & Tree & Bipartite & Planar & Directed Graph    \\ 
\thickhline

\multirow{2}{*}{\begin{tabular}[c]{@{}c@{}}Vertex \end{tabular}}   & \multirow{1}{*}{\begin{tabular}[c]{@{}c@{}}disc\end{tabular}} & \multirow{1}{*}{\begin{tabular}[c]{@{}c@{}}$\Theta(1)$ \end{tabular}} & \multirow{1}{*}{\begin{tabular}[c]{@{}c@{}}$\Theta(1)$ \end{tabular}} & $ O(n^{1/4})$  & $\widetilde{\Theta}(n^{1/4})$  \\

& \multirow{1}{*}{\begin{tabular}[c]{@{}c@{}}herdisc \end{tabular}}  & \multirow{1}{*}{\begin{tabular}[c]{@{}c@{}}$\Theta(1)$ \end{tabular}} & \multirow{1}{*}{\begin{tabular}[c]{@{}c@{}}$\widetilde{\Theta}(n^{1/4})$ \end{tabular}}& $\widetilde{\Theta}(n^{1/4})$  & $\widetilde{\Theta}(n^{1/4})$ \\

\hline

\multirow{2}{*}{\begin{tabular}[c]{@{}c@{}}Edge \end{tabular}}   & \multirow{1}{*}{\begin{tabular}[c]{@{}c@{}}disc\end{tabular}} & \multirow{1}{*}{\begin{tabular}[c]{@{}c@{}}$\Theta(1)$ \end{tabular}} & \multirow{1}{*}{\begin{tabular}[c]{@{}c@{}}$\Theta(1)$ \end{tabular}} & $ O(n^{1/4})$ &  $  \min\left\{O(m^{1/4}), \widetilde{O}(D^{1/2})\right\}$  \\

& \multirow{1}{*}{\begin{tabular}[c]{@{}c@{}}herdisc \end{tabular}}  & \multirow{1}{*}{\begin{tabular}[c]{@{}c@{}}$\Theta(1)$ \end{tabular}} & \multirow{1}{*}{\begin{tabular}[c]{@{}c@{}}$\widetilde{\Theta}(n^{1/4})$ \end{tabular}}  & $\widetilde{\Theta}(n^{1/4})$  &  $ \widetilde{\Omega}(n^{1/4})$ \\

\thickhline

\end{tabular}
}
\end{table}

The upper bound in Theorem \ref{thm:intromain} is constructive and algorithmic; that is, we provide an algorithm that colors vertices (edges, respectively) of the input graph to achieve vertex (edge, respectively) discrepancy $\Oish(n^{1/4})$ on its shortest paths (or on a given subsystem of its shortest paths). Notably, Theorem \ref{thm:intromain} should be contrasted with the fact that the maximum possible discrepancy of any {\em simple path system}\footnote{A path system is \textit{simple} if no individual path repeats nodes.}  of polynomial size in a general graph is $\Thetaish(n^{1/2})$. 
The upper bound of $\Oish(n^{1/2})$ follows by coloring the nodes randomly and applying standard Chernoff bounds.  The lower bound is non-trivial and is proved in \Cref{apdx:discrepancy-path-inconsistent}.
Thus, Theorem \ref{thm:intromain} represents a concrete separation between unique shortest path systems and general path systems.

 We refer to \Cref{tab:res-summary-undirected} for our results for undirected graphs and to \Cref{tab:res-summary} for our results for directed graphs. All the bounds on discrepancy and hereditary discrepancy hold for $\ell_2$-discrepancy and $\ell_2$-hereditary discrepancy. See \Cref{sec:L2} for details.

The main open question that we leave in this work is on the hereditary edge discrepancy of shortest paths in \textit{directed} weighted graphs.
We show the following:
\begin{restatable}{theorem}{digraphopen}
\label{thm:digraph-open}
    For any $n$-node, $m$-edge directed weighted graph $G$ with a unique shortest path between each pair of nodes, the system of shortest paths $\Pi$ satisfies
\[\hdiscv(\Pi) \le O(n^{1/4}) \qquad \text{and} \qquad \hdisce(\Pi) \le O(m^{1/4}).\]
\end{restatable}

Lower bounds in the undirected setting immediately apply to the directed setting as well, and so this essentially closes the problem for directed hereditary \textit{vertex} discrepancy.
It is an interesting open problem whether the upper bound for directed hereditary \textit{edge} discrepancy can be improved to $\Oish(n^{1/4})$ as well.

We also leave open whether our lower bound for hereditary discrepancy extends to (non-hereditary) edge discrepancy as well, and to (non-hereditary) vertex or edge discrepancy of planar graphs.
\vspace{2mm}

\noindent
\textbf{Applications to  Differential Privacy.} 
One application of our discrepancy lower bound on unique shortest paths is in differential privacy (DP)~\cite{dwork2006calibrating, dwork2014algorithmic}. An algorithm is {\em differentially private} if its output distributions are relatively close regardless of whether an individual's data is present in the data set. More formally, for two databases $Y$ and $Y'$ that are identical except for one data entry, a randomized algorithm $\mathcal{M}$ is $(\eps,\delta)$-differentially private if, for any measurable set $A$ in the range of the algorithm $\mathcal M$, 
$$\prob{\mathcal{M}(Y) \in A}\leq e^{\eps} \prob{\mathcal{M}(Y') \in A} +\delta.$$ 

The topic of discrepancy of paths on a graph is related to two problems already studied in differential privacy: \textit{All Pairs Shortest Distances} (APSD)~\cite{ghazi2022differentially,fan2022breaking,sealfon2016shortest} and \textit{All Sets Range Queries}  (ASRQ)~\cite{deng2023differentially}. In both of these problems, we assume that the graph topology is public.  In the APSD problem, the edge weights are not publicly known. A query in APSD is a pair of vertices $(u,v) \in V \times V$ and the answer is the shortest distance between $u$ and $v$. In contrast, in ASRQ problem, the edge weights are assumed to be known, and every edge also has a private \textit{attribute}. Here, the range is defined by the shortest path between two vertices (based on publicly known edge weights). The answer to the query $(u,v) \in V \times V$ then is the sum of private attributes along the shortest path. 
In what follows,  we give a high-level argument for the lower bound on DP-APSD problem; the lower bound of $\Omegaish(n^{1/4})$ for the DP-ASRQ problem also follows nearly the same argument.

Chen \textit{et al.}~\cite{ghazi2022differentially} showed that DP-APSD can be formulated as a linear query problem. In this setting, we are given a vertex-incident matrix $A$ of the ${n\choose 2}$ shortest paths of a graph and a vector $x$ of length $n$, and asked to output $Ax$. They show that the hereditary discrepancy of the matrix $A$ provides a lower bound on the $\ell_{\infty}$ error for any $(\eps, \delta)$-DP mechanism for this problem. With this argument, our new discrepancy lower bound immediately implies the following lower bounds.

\begin{theorem}[Informal version of Corollaries 7.1 and F.1]
\label{thm:intro-dp-apsd-lb}
The $(\eps, \delta)$-DP APSD problem and $(\eps, \delta)$-DP ASRQ problem  require additive error at least $\Omegaish(n^{1/4})$.
\end{theorem}

The best known additive error bound for the DP-ASRQ problem is $\widetilde O(n^{1/4})$~\cite{deng2023differentially}, which, by Theorem~\ref{thm:intro-dp-apsd-lb}, is tight up to a $\text{polylog}(n)$ factor. Prior to this work, the only known lower bounds for DP-ASRQ and DP-APSD  were from a point-line system with hereditary discrepancy of $\Omega(n^{1/6})$~\cite{ghazi2022differentially}.
The best known additive error upper bound for DP-APSD is $\Oish(n^{1/2})$~\cite{ghazi2022differentially,fan2022breaking}.  Closing this gap remains an interesting open problem.

\subsection{Our Techniques}

We provide a brief overview of the techniques used for our upper and lower bounds on discrepancy separately.
\vspace{2mm}

\noindent\textbf{Upper Bound Techniques.}
A folklore structural property of unique shortest paths on undirected graphs is \textit{consistency}.
Formally, a system of paths $\Pi$ in an undirected graph is {\em consistent} if for any two paths $\pi_1, \pi_2 \in \Pi$, their intersection $\pi_1 \cap \pi_2$ is a (possibly empty) contiguous subpath of each.
It is well known that, for any undirected graph $G = (V, E, w)$ with unique shortest paths\footnote{In general, on a weighted undirected graph, one can use random perturbation to ensure all shortest paths are unique.}, its system of shortest paths $\Pi$ is consistent.
An analogous fact holds for a directed graph for paths that go through two vertices $u, v$ in the same precedence order.
Our discrepancy upper bounds actually apply to \textit{any} consistent system of paths -- not just those that arise as unique shortest paths in an undirected graph. 

We use two different proof techniques to obtain upper bounds on discrepancies.
First, we consider the paths as a set system with vertices (for vertex discrepancy) or edges (for edge discrepancy) as the ground set and then apply a standard application of \textit{primal shatter functions} (\Cref{def:shatter-function}), which bounds the number of subsets obtained by limiting to only $s$ elements. For any family of consistent paths in directed or an undirected graph, the primal shatter function is upper bounded by $O(s^2)$.
By a well-known bound on discrepancy through the primal shatter function, this immediately gives us an upper bound of $O(n^{1/4})$ for vertex discrepancy and $O(m^{1/4})$ for edge discrepancy (since edge discrepancy is defined on a ground set of $m$ edges in the graph $G$).

When the graph is dense, this upper bound on edge discrepancy deteriorates, becoming trivial when $m = \Theta(n^2)$. 
We thus present a second proof of $\tilde{O}(n^{1/4})$ for \textit{both} vertex and edge discrepancy for a family of consistent paths in an undirected graph, which explicitly constructs a low-discrepancy coloring. This improves the bounds on vertex discrepancy by polylogarithmic factors and on edge discrepancy by polynomial factors.
The main idea in this construction is to adapt the \textit{path cover} technique, used in the recent breakthrough on shortcut sets \cite{kogan2022new}.
That is, we start by finding a small base set of roughly $n^{1/2}$ node-disjoint shortest paths in the distance closure of the graph. 
These paths have the property that any other shortest path $\pi$ in the graph contains at most $O(n^{1/2})$ nodes that are not in any paths in the base set.
We then color \textit{randomly}, as follows:
\begin{itemize}
\item For every node that is not contained in any path in the base set, we assign its color randomly.
Thus, applying concentration bounds, the contribution of these nodes to the discrepancy of $\pi$ will be bounded by $\pm \Oish(n^{1/4})$. 

\item For every path in the base set, we choose the color of the first node in the path at random, and then alternate colors along the path after that. Then we can argue that by consistency, the nodes in each base path randomly contribute  $+1$ or $-1$ (or $0$) to the discrepancy of $\pi$ (see \Cref{fig:intro-coloring} for a visualization).  Since there are only $n^{1/2}$ paths in the base set, we may again apply concentration bounds to argue that the contribution to discrepancy from these base paths will only be $\pm \Oish(n^{1/4})$.
\end{itemize}

\begin{figure}[htp]
\centering
\includegraphics[width=0.55\linewidth]{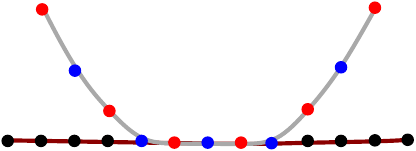}
\caption{If we color the nodes of a unique shortest path with alternating colors, then its nodes will contribute discrepancy $0$, $+1$, or $-1$ to all unique shortest paths that intersect it. }
\label{fig:intro-coloring}
\end{figure}

Summing together these two parts, we obtain a bound of $\Oish(n^{1/4})$ on discrepancy with high probability.
We can translate this to a bound on \textit{hereditary} discrepancy using the fact that consistency is a hereditary property of path systems. 

\medskip
\noindent 
\sloppy
\textbf{Lower Bound Techniques.} 
Lower bounds on discrepancy are typically obtained by the trace bound (e.g., by Chazelle and Lvov~\cite{Chazelle2000-ld}) on an explicit graph construction. The state-of-the-art lower bounds on the discrepancy of unique shortest paths were achieved using a point-line construction of Erd\H{o}s~\cite{Pach2011-nq}, which had $n$ points and $n$ lines in $\mathbb{R}^2$ with each point staying on $\Theta(n^{1/3})$ lines and each line going through $\Theta(n^{1/3})$ points.
This point-line system also implies tight lower bounds for the Szemer{\' e}di-Trotter theorem~\cite{Szemeredi1983-tv} and the discrepancy of arrangements of lines in the plane \cite{Chazelle2000-ld}. It can be associated with a graph that possesses useful properties derived from geometry. If edges in this graph are weighted by Euclidean distance, then the paths in the graph corresponding to straight lines are unique shortest paths by design. Two such shortest paths (along straight lines) only intersect at most once. 

Probably not a coincidence, this point-line construction also provides lower bounds on the graph hopset problem. An (exact) hopset of a graph $G$ with hopbound $\beta$ is a small set of additional edges $H$ in the distance closure of $G$, such that every pair of nodes has a shortest path in $G \cup H$  containing at most $\beta$ edges. Until recently, the state-of-the-art lower bound on the size of the hopset uses exactly the graph derived from the point-line incidence example.
Recently, a construction by Bodwin and Hoppenworth~\cite{bodwin2023folklore} obtained stronger hopset lower bounds with a different geometric graph construction, which still took place in $\rr^2$ but  allowed shortest paths to have many vertices/edges in common.
We show that this construction can be repurposed to derive a stronger lower bound of $\Omegaish(n^{1/4})$ on vertex hereditary discrepancy by applying the trace bound of Larsen~\cite{larsen2017constructive}.  
Combined with our upper bounds, this substantially improves our understanding of the discrepancy of unique shortest paths.

The above upper and lower bounds are for general graphs.
Naturally, one can ask if we have better bounds for special families of graphs.
We further show that the lower bounds remain the same for two interesting families: {\em planar graphs} and {\em bipartite graphs}.
The lower bound construction mentioned above is not planar, and so this requires some additional work.
A natural attempt is to restore planarity by adding vertices to the construction wherever two edges cross.
However, this comes at a cost of an increase in the number of vertices and also with a potential danger of altering the shortest paths.
We show that the number of crossings is not too much higher than $n$.
Then, by carefully changing the weights of the edges and by exploiting the geometric properties of the construction, we show that the topology and incidence of shortest paths are not altered.
For bipartite graphs, although the vertex discrepancy can be made very low -- by coloring the vertices on one side $+1$ and vertices on the other side $-1$ -- the hereditary discrepancy can be as high as the general graph setting. Specifically, we show that the lower bound graph $G$ can be ``lifted'' to a bipartite graph, which essentially preserves the same hereditary discrepancy. 

\section{Related Work}

\subsection{Discrepancy on graphs}
\label{apdx:graph-disc}
In graph theory, the {\em discrepancy} of a graph introduced by Erd{\H o}s~\cite{erdHos1963combinatorial} is defined as follows: 
\begin{align*}
    \max_{U\subset V} \left|e(G[U])-p{{|U|} \choose 2}\right|,
\end{align*}  where $e(G[U])$ is the number of edges of the induced subgraph $G[U]$ on vertices $U \subseteq V$ and $p={|E| \over {n \choose 2}}$ is the density of edges. If we consider a complete graph and randomly color each edge with probability $p$, the above definition of discrepancy quantifies the deviation of induced subgraphs of $G$ from their expected size. 
Erd{\H o}s and Spencer~\cite{erdiis1972imbalances} showed that the graph discrepancy is $\Theta(n^{3/2})$ when $p=1/2$.
This definition and related definitions (e.g., positive discrepancy, dropping the absolute operator) have applications to quasi-randomness~\cite{chung1989quasi}, graph cuts and edge expansion~\cite{alon1997edge,bollobas2006discrepancy,raty2023positive}. There is also study of multicolor discrepancy~\cite{freschi2021note,gishboliner2022discrepancies} that we skip here.

Of particular relevance to our work, Balogh et al.~\cite{balogh2020discrepancies} studied edge discrepancy (as defined in this paper) of (spanning) trees, paths and Hamilton cycles of a graph $G$. In particular, they showed that, for any labeling (i.e., coloring) of edges of $G$, there is a path with discrepancy $\Omega(n)$, even when the graph is a grid.  Prior to this, either probabilistic construction exhibiting such a lower bound was known~\cite{beck1983size, dudek2015alternative,letzter2016path} or an explicit construction of linear size non-planar graphs was known~\cite{alon1988explicit}. The construction for the planar graph in Balogh et al.~\cite{balogh2020discrepancies} can be extended to coloring of vertices such that there is a path with vertex discrepancy $\Omega(n)$ when there are exponentially many paths and $\Omega(\sqrt{n})$ when there are polynomially many paths. (see \Cref{apdx:discrepancy-path-inconsistent} for details).

{Discrepany of paths in directed graphs has also been studied. Reimer~\cite{reimer2002ramsey} showed that, if a directed graph has discrepancy $\Omega(n)$, then the graph must have $\Omega(n^2)$ edges. In the case when we do not allow antiparallel edges, Ben-Eliezer et al.~\cite{ben2012size} showed that there is a directed graph with $\Theta(k^2 \log^2 k)$ edges such that any mapping $\chi: E \to \{-1,1\}$ will either have a path of length $\Omega(k)$ and all edges mapped to $-1$ or a path of length $\Omega(k\log k)$ with all edges mapped to $+1$.}

\subsection{Connection with curve discrepancy}
A classical topic in computational geometry is to study upper and lower bounds of the discrepancy/hereditary discrepancy of the incidence matrix of geometric objects and a set of points. For example, for a set of $n$ points and $n$ halfplanes in $\reals^2$, the $n$ by $n$ incidence matrix (with rows corresponding to halfplanes and  columns corresponding to points) has discrepancy of $
\Omega(n^{1/4})$. For $n$ points and $n$ lines in the plane, the discrepancy of the incidence matrix is $\Omega(n^{1/6})$~\cite{Chazelle2000-ld}. In general the discrepancy of such incidence matrix is related to the `complexity' of the geometric shapes. In our setting of a graph, the set of all pairs shortest paths defines a set system on the vertices. When the graph is planar, the shortest paths are essentially simple curves in the plane.

We would like to compare our results with discrepancy of curves and points in the geometric setting.
Using the classification in Pach and Sharir~\cite{Pach1998-ib}, a family of simple curves have $k$ \emph{degree of freedom} and \emph{multiplicity type} $s$
if, for any $k$ points, there are at most $s$ curves passing through all of them, and any pair of curves intersect in at most $s$ points. Lines in the plane have degree of $2$ and multiplicity of $1$. A set of curves with degree of $2$ and multiplicity of $1$ is called \emph{pseudolines} -- two pseudolines have at most one intersection. 
For $n$ points and $n$ lines, the discrepancy is upper bounded by $O(n^{1/6}(\log n)^{2/3})$~\cite{Chazelle2000-rd} which is nearly tight by a polylogarithmic factor. The proof uses the standard partial coloring argument with the Szerem\'{e}di-Trotter bound on point-line incidence -- for any $n$ points and $m$ lines there are at most $O(m^{2/3}n^{2/3}+m+n)$ point-line incidences~\cite{Szemeredi1983-tv}. The Szerem\'{e}di-Trotter bound can be extended to a set of pseudolines~\cite{Pach1998-ib,Szekely1997-yp}. Therefore the same proof and upper bound hold for the discrepancy of pseudolines. 

For a consistent set of shortest paths, two shortest paths will only intersect at a contiguous segment, which may have multiple vertices/points. Thus using the curve classification criterion, a consistent family of shortest paths in the plane has degree of $2$ but multiplicity $s$ that is possibly higher than a constant. In fact, our discrepancy lower bound construction in the planar graph setting uses a design with $s$ possibly as high as $n^{1/2}$. This is the major difference of shortest paths in a planar graph with pseudolines, which allows the discrepancy of shortest paths to go beyond the pseudoline upper bound of $\Oish(n^{1/6})$\footnote{Using the incidence upper bound for $k=2$ and $s=n^{1/2}$ from~\cite{Szekely1997-yp} and partial coloring, one can obtain a discrepancy upper bound of $\Oish(n^{1/3})$ for our path construction. In contrast, we obtained a nearly tight bound of $\Thetaish(n^{1/4})$.}. 

\section{Preliminaries}

A \emph{path system} is a pair $S = (V, \Pi)$ where $V$ is a ground set of nodes and $\Pi$ is a set of vertex sequences called \textit{paths}. Each path may contain, at most, one instance of each node. 
We now formally define consistency,  a structural property of unique shortest paths that will be useful. 

\begin{definition}
A path system $S = (V, \Pi)$ is {\em consistent} if no two paths in $S$ intersect, split apart, and then intersect again later. Formally:
\begin{itemize}
\item In the undirected setting, consistency means that for all $u, v \in V$ and all $\pi_1, \pi_2 \in \Pi$ such that $u, v \in \pi_1 \cap \pi_2$, we have that $\pi_1[u, v] = \pi_2[u, v]$, i.e., the intersection of $\pi_1$ and $\pi_2$ is a \textit{contiguous} subpath (subsequence) of $\pi_1$ and $\pi_2$. 

\item In the directed setting, consistency means that for all $u, v \in V$ and all $\pi_1, \pi_2 \in \Pi$ such that $u$ \textbf{precedes} $v$ in both $\pi_1 $ and $\pi_2$, we have that $\pi_1[u, v] = \pi_2[u, v]$.
\end{itemize}
\end{definition}

In every weighted graph for which all pairs shortest paths exist (i.e., no negative cycles), we can represent all-pairs shortest paths using a consistent path system. 
In particular, if all shortest paths are \textit{unique}, then consistency is implied immediately. 

We will investigate the combinatorial discrepancy of path systems $(V, \Pi)$.  Usually, we will assume that $|V| = n$ and $|\Pi|$ is polynomial in $n$. We define a vertex coloring $\chi: V \mapsto \{-1, 1\}$ and define the \textit{discrepancy} of $\Pi$ as
\[
\disc(\Pi) = \min_{\chi} \chi(\Pi),\, \quad \mbox{where } \chi(\Pi)=\max_{\pi \in \Pi} \left| \chi(\pi)\right|,\, \chi(\pi)=\sum_{v \in \pi} \chi(v). 
\]

Using a random coloring $\chi$, we can guarantee with high probability that for all paths $\pi \in \Pi$~\cite{Chazelle2000-rd}:  $$|\chi(\pi)|\leq \sqrt{2|\pi|\ln (4|\Pi|)}.$$
This immediately provides a few observations. 

\begin{observation}
\label{obs:random-coloring-prelims}
When $\Pi$ is a set of paths with size polynomial in $n$, then $\disc(\Pi)=O(\sqrt{n\log n})$. This bound is true even for paths that are possibly non-consistent. \end{observation}
\begin{observation}
 When the longest path in $\Pi$ has $D$ vertices we have $\disc(\Pi)=O(\sqrt{D\log n})$. Thus, for graphs that have a small diameter (e.g., small-world graphs), the discrepancy of shortest paths is automatically small. 
\end{observation}

\textit{Hereditary discrepancy} is a more robust measure of the complexity of a path system $(V, \Pi)$, defined as $\herdisc(\Pi)=\max_{Y\subseteq V} \disc(\Pi|_Y)$, where $\Pi|_Y$ is the collection of sets of the form $\pi \cap Y$ with $\pi\in \Pi$.  Clearly, $\herdisc(\Pi)\geq \disc(\Pi)$. Sometimes the discrepancy of a set system may be small while the hereditary discrepancy is large~\cite{Chazelle2000-rd}. Thus, in the literature, we often talk about lower bounds on the hereditary discrepancy.

Now that we have defined vertex and edge (hereditary) discrepancy, one may wonder if there is an underlying relationship between vertex and edge (hereditary) discrepancy since they share the same bounds in most settings presented in \Cref{tab:res-summary}. The following observation shows that vertex discrepancy bounds directly imply bounds on edge discrepancy. 
\begin{observation}
\label{obs:vertex-implies-edge}
Denote by $\disc(n)$ (and $\herdisc(n)$) the maximum discrepancy (maximum hereditary discrepancy, respectively) of a consistent path system of a (undirected or directed) graph of $n$ vertices.
We have that
\begin{enumerate}
\item Let $g(x)$ be a non-decreasing function. If $\hdiscv(n)\geq g(n)$, then $$\hdisce(n) \geq g(n/2).$$
\item Let $f(x)$ be a non-decreasing function. If $\discv(n) \leq f(n)$, then $$\disce(m) \leq f(m).$$
\end{enumerate}
\end{observation}

\begin{proof}
We first show that if an undirected graph $G=(V,E,w)$ with the consistent path matrix $A_v$ has hereditary discrepancy at least $g(n)$, we can obtain another undirected graph $G'=(V', E', w')$ and matrix $A_e$ as the (consistent path) edge incidence matrix with hereditary discrepancy at least $g(n/2)$. The construction is as follows. 
\begin{enumerate}
    \item [(a)] We first split each vertex $v\in V$ in $G$ to two vertices $(v_{\mathsf{in}}, v_{\mathsf{out}})$ to obtain $V'$.
    \item [(b)] For every $v\in V$, add a single edge $(v_{\mathsf{in}}, v_{\mathsf{out}})$ to $E'$ {with weight $0$}.
    \item [(c)] For any $v\in V$ and each edge $(u,v)\in E$ (with the fixed $v$), add edges $(u_{\mathsf{out}}, v_{\mathsf{in}})$ and $(u_{\mathsf{in}}, v_{\mathsf{out}})$ {with weights $w'((u_{\mathsf{out}}, v_{\mathsf{in}})) = w'(u_{\mathsf{in}}, v_{\mathsf{out}})= w(u,v)$} to $E'$.
\end{enumerate}

The path incident matrix $A_e$ is defined as follows: for each path as a row $a$ of $A_v$, construct a new path in $G'$ by following the order of $u_{\mathsf{out}} \rightarrow v_{\mathsf{in}}\rightarrow v_{\mathsf{out}} \rightarrow w_{\mathsf{in}}$ for a $u \rightarrow v \rightarrow w$ sequence. For each row in $A_v$, if the element corresponding to vertex $u$ is a $1$, we mark the edge $(u_{\mathsf{in}}, u_{
\mathsf{out}})$ as $1$ in $A_e$. Note that the new path system defined by $A_e$ remains consistent: for any intersection between the two paths $P_1 \cap P_2 = (u_{1}, u_{2}, \cdots, u_{\ell})$, the intersection remains a single path of $(u_{1, \mathsf{in}}, u_{1, \mathsf{out}},\cdots, u_{\ell,\mathsf{in}}, u_{\ell, \mathsf{out}})$ in $A_e$.

Let $Y$ be the columns that induces the $g(n)$ discrepancy on $G$, i.e., \[\min_{x\in \{-1,+1\}^{\card{Y}}}\norm{A_{v_{Y}} x}_{\infty} = g(n).\] 
Now, observe that, for each row in $A_e$, an edge $(v_{\mathsf{in}}, v_{\mathsf{out}})$ is marked as $1$ if and only if $v$ is marked as $1$ in $A_v$. Therefore, we can take the new set $Y'$ as the edges corresponding to $Y$, and there is 
\[\min_{x\in \{-1,+1\}^{\card{Y'}}}\norm{A_{e_{Y'}} x}_{\infty} = \min_{x\in \{-1,+1\}^{\card{Y}}}\norm{A_{v_{Y}} x}_{\infty} = g(n).\] 
Finally, since graph $G'$ has $n' = 2n$ vertices, we have the hereditary discrepancy to be at least $g(n'/2)$, as desired.

For directed graphs, for each vertex $u$ in $G$ we create a directed edge $(u_{\mathsf{in}}, u_{
\mathsf{out}})$ in $G$. If there is a directed edge $(u, w)$ in $G$, we create a directed edge $(u_{\mathsf{out}}, w_{
\mathsf{in}})$. The rest of the proof follows.

{
We next show that if the vertex discrepancy of a consistent path system on a graph of $n$ vertices is at most $f(n)$, the edge discrepancy of graphs of $m$ edges is at most $f(m)$. For a graph $G$ with $n$ vertices, $m$ edges, and a path incident matrix $A_e$, we create a companion graph $G'$ in which we create a vertex $v_{e}$ for each $e\in G$ and connect two vertices $v_e$ and $v_{e'}$ if $e$ and $e'$ are incident to the same vertex (share a vertex) in $G$. 
A path in $G$ that visits a sequence of paths $e_1, e_2, \cdots e_k$ maps to a path that goes through vertices $v_{e_1}, v_{e_2}, \cdots, v_{e_k}$ in $G'$.
By the consistency of $A_e$, the new path incident matrix $A_v$ for $G'$ also characterizes a consistent path system.
Thus, since the vertex discrepancy of $G'$ is at most $f(m)$ where $m$ is the number of vertices in $G'$, the 
the edge discrepancy of $G$ is at most $f(m)$ as well by the above construction.
}

Finally, if $G$ is a directed graph, we connect a directed edge from $v_e$ to $v_{e'}$ only when the tail vertex of $e$ is the same as the head vertex of $e'$. Everything else follows. 
\end{proof}

We also use some technical tools from discrepancy theory and statistics. 

\subsection{Known Results in Discrepancy Theory}
The first result that we discuss is the one that gives an upper bound on the discrepancy of a set system in terms of {\em primal shatter function}.

\begin{definition}[Primal Shatter Function]
\label{def:shatter-function}
Let $(X,\R)$ be a set system, i.e., $X$ is a ground set and $\R=\{S_1, S_2, \cdots, S_\ell\}$ with $S_i \subseteq X$ for all $1\leq i \leq \ell$. Let
 $s$ be a positive integer. The \emph{primal shatter function}, denoted as $\pi_{\R}(s)$, is defined as 
 \[
 \pi_{\R}(s):=\max_{A \subseteq X: \, |A|=s}\card{\{A\cap S \mid S\in\R\}}.
 \]
\end{definition}

The following is a well-known result of the discrepancy theory.

\begin{proposition}[Theorem 1.2 in Matousek~\cite{Matousek1995-as}]
\label{thm:shatter-function}
Given a set system $(X,\R)$, the discrepancy of a range space $\R$ whose primal shatter function is bounded by $\pi_{\R}(x)=cx^d$, for some constant $c>0$, $d>1$, is 
\[O(n^{1/2-1/(2d)}),\]
where $n$ is the size of the ground set, and $O(\cdot)$ hides the polynomial dependency on $c$ and $d$. 
\end{proposition}

For the lower bound on the hereditary discrepancy, one general result is the {\em trace bound} first shown by Chazelle and Lyov~\cite{Chazelle2000-ld}. 

\begin{lemma}[Trace Bound of Chazelle and Lyov~\cite{Chazelle2000-ld}]
\label{lem:trace-bound}
    If $A$ is an $m$ by $n$ incidence matrix and $M=A^\top A$, then there is an absolute constant $0<c<1$ such that 
    \[\herdisc(A)\geq \frac{1}{4} c^{n\cdot  \tr(M^2)/(\tr(M))^2}\sqrt{\frac{\tr(M)}{n}}.\] 
\end{lemma}

Recently, this trace bound has been improved in the exponential term through a series of works culminating in the following bound in Larsen~\cite{larsen2017constructive}, which we also use in our work:
\begin{lemma}
[Trace Bound of Larsen~\cite{larsen2017constructive}]
\label{lem:larsen}
    If $A$ is an $m$ by $n$ incidence matrix and $M=A^\top A$, then \[\herdisc(A)\geq {(\tr(M))^2 \over 8e\tr(M^2) \cdot \min\{m,n\}}\sqrt{\frac{\tr(M)}{\max\{m,n\}}}.\] 
\end{lemma}

We give various interpretations of $\tr(M)$ and $\tr(M^2)$ in \Cref{lem:larsen} that would be useful later on. 
Algebraically, $\tr(M)$ is the sum of its eigenvalue while $\tr(M^2)$ is the sum of the square of the eigenvalues. Combinatorially, $\tr(M)$ is the number of ones in $A$ and $\tr(M^2)$ is the number of rectangles whose four corners are ones in $A$. Geometrically, $\tr(M)$ is the count of point/region incidences, and $\tr(M^2)$ is the number of pairs of points in all the pairwise intersections of regions. Finally, if $A$ is the incidence matrix of a bipartite graph $G$ where one set of vertices correspond to the rows and the other set of vertices correspond to the columns, $\tr(M^2)$ is the number of length $4$ cycles in $G$. Based on the algebraic interpretation, it means that the trace bound is non-trivial whenever all the eigenvalues of $A$ are fairly uniform. This can be seen by noticing that, if $\{\lambda_1, \cdots, \lambda_n\}$ are eigenvalues of $A$,  then 
$\tr(M)^2  = n \cos^2(\theta) \tr(M),$ 
where $\theta$ is the angle between the vector $(\lambda_1, \cdots, \lambda_n)$ and the all one-vector.

Our lower bound construction requires a hard instance of a class of graphs. For that, we use the construction of Bodwin and Hoppenworth \cite{bodwin2023folklore}, whose key properties that we use are stated as the following lemma:
\begin{lemma}[Lemma 1 of Bodwin and Hoppenworth \cite{bodwin2023folklore}]
\label{lem:hopset_lemma}
For any $p \in [1, n^2]$, there is an infinite family of $n$-node undirected weighted graphs $G = (V, E, w)$ and sets $\Pi$ of $p$ paths in $G$ such that
\begin{itemize}
\item $G$ has $\ell = \Theta\left(\frac{n}{\sqrt{p \log n}} \right)$ layers. Each path in $\Pi$ starts in the first layer, ends in the last layer, and contains exactly one node in each layer. 

\item Each path in $\Pi$ is the unique shortest path between its endpoints in $G$.

\item For any two nodes $u, v \in V$, there are at most $ \frac{\ell}{h(u, v)} $ paths in $\Pi$ that contain both $u$ and $v$,  where $h(u, v)$ is the hop-distance (number of edges on the shortest path) between $u$ and $v$ in $G$ and $1 \leq h(u, v) \leq \ell$. 

\item Each node $v \in V$ lies on at most $O\left(\frac{\ell p}{n}\right)$ distinct paths in $\Pi$. 
\end{itemize}
\end{lemma}

\subsection*{Concentration Inequalities}
We use the following standard variants of the Chernoff-Hoeffding bound in our paper. 

\begin{proposition}[Chernoff bound]\label{prop:chernoff}
	Let $X_1,\ldots,X_n$ be $n$ independent random variables with support on $\{0,1\}$. Define $X := \sum_{i=1}^{n} X_i$. Then, for every $\delta >0$, there is 
	\[
	\Pr \left[X\geq (1+\delta)\cdot \expect{X}\right] \leq \exp\paren{-\frac{\delta^2}{\delta+2}\cdot \expect{X}}.
	\]
	In particular, when $\delta\in (0,1]$, there is
	\[
		\Pr[\card{X - \expect{X}} > \delta\cdot \expect{X}] \leq 2 \cdot \exp\paren{-\frac{\delta^{2} \expect{X}}{3}}. 
	\]
\end{proposition}

\begin{proposition} [Additive Chernoff bound]\label{prop:additive-chernoff}
	Let $X_1,\ldots,X_n$ be $n$ independent random variables with support in $[0,1]$. Define $X := \sum_{i=1}^{n} X_i$. Then, for every $t>0$,
	\[
	    \Pr[\card{X - \expect{X}}>t] \leq 2\cdot 
	    \exp\paren{-\frac{2t^2}{n}}.
	\]
\end{proposition}

\section{General Graphs: Upper Bound Existential Proof}
\label{subsec:vert_disc_ex}

This section collects the existential proof of the upper bounds on vertex- and edge-discrepancy for consistent path systems in (possibly) directed graphs. 
Our approach uses \Cref{thm:shatter-function}, which gives a discrepancy upper bound using the primal shatter function of a set system. This approach leads to the same upper bounds for undirected and directed graphs. In specifics, we show an upper bound of $O(n^{1/4})$ holds for vertex discrepancy, while the edge discrepancy is at most $O(m^{1/4})$. That is, we show an existential proof of \Cref{thm:digraph-open} (on directed graphs) by this approach. Note that for undirected graphs, we can achieve better edge discrepancy bounds using explicit constructions in \Cref{apdx:edge-disc-exp}. 

\digraphopen*
\begin{proof}
    We consider vertex discrepancy first. Let $S = (V, \Pi)$ be the path system containing all $|\Pi| = {n \choose 2}$ unique shortest paths in $G$ over vertex set $V$. We can interpret $S$ as a set system (e.g., by ignoring the ordering of vertices in paths $\pi \in \Pi$).
    
    We claim that the primal shatter function $\pi_S$ of $S$ is $\pi_S(x) = O(x^2)$. The intersection of any set $A \subseteq V$ of $|A| = x$ vertices with a path $\pi = \pi[s, t] \in \Pi$ is equal to $A \cap \pi[u, v]$ with $u$ and $v$ being the first and last vertex on path $\pi$ in set $A$, respectively. Then we have
    $$
    |\{A \cap \pi \mid \pi \in \Pi\}| \leq |A|^2 = x^2, 
    $$
    and $\pi_S(x) = O(x^2)$, as claimed.  
    Since the size of the ground set is $|V| = n$, \Cref{thm:shatter-function} implies that the (non-hereditary) vertex discrepancy of the incidence matrix for a family of consistent paths on an $n$-node graph is at most $O(n^{1/4})$. An upper bound of $O(m^{1/4})$ for edge discrepancy on $m$-edge graphs follows from Observation \ref{obs:vertex-implies-edge}. 

    Finally, to show the upper bound on  \emph{hereditary discrepancy}, we observe that
    for any subset $U\subseteq V$, we can define the system $\Pi[U]$ of the paths in $\Pi$ induced on the nodes in $U$. This path system $\Pi[U]$ is also consistent. Applying the above argument on $\Pi[U]$ therefore give us an $O(n^{1/4})$ upper bound for the discrepancy of $\Pi[U]$, implying our desired vertex hereditary discrepancy upper bound. A similar argument achieves an edge hereditary discrepancy upper bound of $O(m^{1/4})$.
\end{proof}

Notice that for a sparse graph (e.g., $m=O(n)$) the upper bound on edge discrepancy matches the upper bound on vertex discrepancy, but for a dense graph (e.g., $m=\Theta(n^2)$), the upper bound becomes $O(n^{1/2})$, which is no better than the upper bound by random coloring.

\section{Undirected Graphs: Lower Bound and Explicit Colorings} \label{sec:genundir}

We now discuss the main result for undirected graphs (\Cref{thm:intromain}). We first show in \Cref{subsec:undirected_lb} a hereditary discrepancy lower bound of $\Omega(n^{1/4}/\sqrt{\log n})$ for both edge and vertex discrepancy in general undirected graphs. Then, in \Cref{subsec:vert-disc-exp}, we present a vertex coloring achieving hereditary discrepancy of $\Oish(n^{1/4})$. Finally, we present an explicit edge coloring with the same hereditary discrepancy bound in \Cref{subsec:edge-disc-exp}.

\subsection{Lower Bound}
\label{subsec:undirected_lb}

As suggested by \Cref{obs:vertex-implies-edge}, we focus on the vertex hereditary discrepancy.
We then show in \Cref{apdx:disc-equal-herdisc} that this theorem implies the same lower bound on (non-hereditary) vertex discrepancy as well.

\begin{theorem}
\label{thm:undi-hdisc-lb}
    There are examples of $n$-vertex undirected weighted graphs $G$ with a unique shortest path between each pair of vertices in which this system of shortest paths $\Pi$ has
$$\hdiscv(\Pi) \ge \Omega\left({n^{1/4}\over\sqrt{\log(n)}}\right). $$
\end{theorem}

To obtain the lower bound, we employ the new graph construction by Bodwin and Hoppenworth \cite{bodwin2023folklore}, which shows that any exact hopset with $O(n)$ edges must have at least $\Omegaish(n^{1/2})$ hop diameter. Despite seeming unrelated, this construction also sheds light on our problem. Another technique we use to show the hereditary discrepancy lower bound is the trace bound \cite{larsen2017constructive}. In the following proof section, we first summarize the construction related to our objective, then show the calculation using the trace bound that leads to our lower bound.

\begin{proof}
    The key properties of the graph construction in Bodwin and Hoppenworth~\cite{bodwin2023folklore} that we need can be summarized in \Cref{lem:hopset_lemma}. We will make use of the shortest path vertex incidence matrix of the graph in Bodwin and Hoppenworth~\cite{bodwin2023folklore}.
Recall that hereditary discrepancy considers the sub-incidence matrix induced by columns corresponding to a set of vertices. We select the set of vertices occurring in the paths in  $\Pi$, and show it leads to hereditary discrepancy at least $\Omega(n^{1/4}/\sqrt{\log n})$. Specifically, take $A$ as the incidence matrix so each row corresponds to one path in $\Pi$. $A$ has dimension $p \times n$ where $n$ is the number of vertices in $G$ and the $(i,j)$-th entry of $A$ is $1$ is the vertex $j$ is in the path $i$. 

Now define $M=A^\top A$. Recall that $\tr(M)$ is the number of $1$s in the matrix $A$. Since by construction, every path has length $\ell$, we have $\tr(M)=p\ell$.  
Furthermore, let $m_{ij}$ be the $(i, j)$-th element of matrix $M$, and observe that it is exactly the number of paths that contain vertices $i$ and $j$. Note that $m_{ij}=m_{ji}$. Additionally,  $\tr(M^2)$ is the number of length 4 closed walks in the bipartite graph representing the incidence matrix $A$. This implies that  
\begin{align}
\begin{split}
    \tr(M^2) &= \sum_{j=1}^{p}  \sum_{\substack{u, v\in P_j,\\ u\neq v}} m_{u, v} + \sum_{i=1}^{n} m_{ii}^2 
    = \sum_{j=1}^{p} \sum_{i=1}^{\ell} \sum_{\substack{u, v\in P_j,\\ h(u, v)=i}} m_{u, v}+n\cdot O\left(\frac{p\ell}{n}\right)^2\\
    &\leq \sum_{j=1}^{p} \sum_{i=1}^{\ell} \ell \cdot \frac{\ell}{i} +O\left(\frac{p^2\ell^2}{n} \right) \leq  p\ell^2 \log \ell+O\left(\frac{p^2\ell^2}{n} \right) .
\end{split}
\label{eq:trace_Msquared}
\end{align}

By setting $p=n\log n$,  it follows that $\ell= \Theta\left({\sqrt{n}\over\log n} \right)$. 
Further, 
\begin{align*}
n p \ell^2 \log\ell \leq O\left(n \cdot n\log(n) \cdot \frac{n}{\log^{2}{n}}\cdot \log(n)\right) = O(n^3) = O(p^2 \ell^2).
\end{align*}

By equation~\cref{eq:trace_Msquared}, we have $\tr(M^2) =O(p^2\ell^2/n)$. Using this and  $\tr(M)=p\ell$ in the trace bound in \Cref{lem:larsen}~\cite{larsen2017constructive} gives us
\begin{align*}
\herdisc_v(\Pi) & \geq  {(\tr(M))^2 \over 8e\min\{p,n\} \cdot\tr(M^2)} \sqrt{\frac{\tr(M)}{\max\{p,n\}}}\\ 
	&= {(\tr(M))^2 \over 8en \cdot\tr(M^2)} \sqrt{\frac{\tr(M)}{p}}  \geq \Omega\left({p^2\ell^2 \over p^2 \ell^2} \sqrt{\ell}\right) \\
	&= \Omega(\sqrt{\ell})= \Omega\left(\frac{n^{1/4}}{\sqrt{\log(n)}}\right)
\end{align*}
The completes the proof of the lower bound.
\end{proof}

\subsection{Vertex Discrepancy Upper Bound -- Explicit Coloring}
\label{subsec:vert-disc-exp}
In this subsection, we provide an explicit upper bound on the vertex discrepancy $\chi(\Pi)$ of a consistent path system $(V, \Pi)$ with $|V| = n$ and $|\Pi| = \textrm{poly}(n)$ when the underlying graph is undirected. 

\begin{theorem}
    For a consistent path system $S = (V, \Pi)$ in an undirected graph where $|V| = n$ and $|\Pi| = \textrm{\normalfont poly}(n)$,  there exists a labeling $\chi$ such that $\chi(\Pi) = O(n^{1/4}\log^{1/2}(n)).$
    \label{thm:disc_upper_vertex}
    Consequently, every $n$-vertex undirected graph has hereditary  vertex discrepancy $O(n^{1/4} \log^{1/2}(n))$.
\end{theorem}

Let $S = (V, \Pi)$ be a consistent path system with $|V| = n$ and $|\Pi| = \textrm{poly}(n)$. As the first step towards constructing our labeling $\chi: V \mapsto \{-1, 1\}$, we will construct a collection of paths $\Pi'$ on $V$ that will have a useful covering property, stated in \Cref{prop:path_cover}, over the paths in $\Pi$. 

\vspace{2mm}

\noindent
\textbf{Constructing path cover $\Pi'$.} 
Initially, we let $\Pi' = \emptyset$. We define $V'$ to be the set of all nodes in $V$ belonging to a path in $\Pi'$, i.e., $V' := \bigcup_{\pi' \in \Pi'} \pi'.$ While $\left|\pi \setminus V' \right| \geq n^{1/2}$ for some $\pi \in \Pi$, find a (possibly non-contiguous) subpath of $\pi$ of length $n^{1/2}$ that is vertex-disjoint from all paths in $\Pi'$. Formally, find a subpath $\pi' \subseteq \pi$ such that $|\pi'| = n^{1/2}$ and $\pi' \cap V' = \emptyset$. Add path $\pi'$ to path cover $\Pi'$ and update $V'$. Repeatedly add paths to path cover $\Pi'$ in this manner until $|\pi \setminus V'| < n^{1/2}$ for all $\pi \in \Pi$.

\begin{proposition} 
\label{prop:path_cover} 
Path cover $\Pi'$ satisfies the following properties:
\begin{enumerate}
    \item for all $\pi \in \Pi'$, $|\pi| = n^{1/2}$, 
    \item the number of paths in $\Pi'$ is $|\Pi'| \leq n^{1/2}$, 
    \item {\em (Disjointness Property)}  The paths in $\Pi'$ are pairwise vertex-disjoint,
    \item {\em (Covering Property)} For all $\pi \in \Pi$, the number of nodes in $\pi$ that do not lie in any path in path cover $\Pi'$ is at most $n^{1/2}$. Formally, let $V' = \cup_{\pi' \in \Pi'} \pi'$. Then 
    $\forall \pi \in \Pi, 
    \left|\pi \setminus V' \right| \leq n^{1/2},
    $
    \item {\em (Consistency Property)} 
    {
    For all $\pi \in \Pi$ and $\pi' \in \Pi'$, the intersection $\pi \cap \pi'$ is a (possibly empty) contiguous subpath of $\pi'$.\footnote{Note that it may not be true that $\pi \cap \pi'$ is a contiguous subpath of $\pi$.} That is, for any two distinct vertices $u, w \in \pi\cap\pi'$, a vertex $v\in \pi'$ in between $u, w$ on $\pi'$ must be on $\pi\cap\pi'$ as well.}
\end{enumerate}
\end{proposition}
\begin{proof}
    Properties 1, 3, and 4 follow from the construction of $\Pi'$. Property 2 follows from Properties 1 and 3 and the fact that $|V| = n$. The Consistency Property of $\Pi'$ is inherited from the consistency of path system $S$. 
    Specifically, by the construction of $\Pi'$, path $\pi' \in \Pi'$ is a subpath of a path $\pi'' \in \Pi$.
    Recall that by the consistency of path system $S$, the intersection $\pi \cap \pi''$ is a (possibly empty) contiguous subpath of $\pi''$. 
    {
    Suppose towards contradiction that $\pi \cap \pi'$ is \textit{not} a contiguous subpath of $\pi'$. Then there exists three distinct vertices $u, v, w \in \pi'$ such that $v$ is between $u$ and $w$ in $\pi'$ (i.e., $v \in \pi'[u, w]$); $u, w \in \pi' \cap \pi$; and $v  \in \pi' \setminus \pi$. Since $\pi' \subseteq \pi''$, this implies that $(\pi \cap \pi'')[u, w]$ is not a contiguous subpath of $\pi''$, contradicting the consistency of path system $S$.     
    We conclude that $\pi \cap \pi'$ is a contiguous subpath of $\pi'$.}
\end{proof}

\noindent
\textbf{Constructing labeling $\chi$.}
Let $\pi' = (v_1, \dots, v_k) \in \Pi'$ be a path in our path cover. We will label the nodes of $\pi'$ using the following random process.  With probability $1/2$ we define $\chi: \pi' \mapsto \{-1, 1\}$ to be 
\[
\chi(v_i) = 
\begin{cases} 
      1 &   i \equiv 0 \mod 2  \text{ and } i \in [1, k] \\
      -1 &  i \equiv 1 \mod 2  \text{ and } i \in [1, k] 
   \end{cases},
\]
and  with probability $1/2$ we define $\chi: \pi' \mapsto \{-1, 1\}$ to be 
\[
\chi(v_i) = 
\begin{cases} 
      -1 &   i \equiv 0 \mod 2  \text{ and } i \in [1, k] \\
      1 &  i \equiv 1 \mod 2  \text{ and } i \in [1, k] 
   \end{cases}.
\]
The labels of consecutive nodes in $\pi'$ alternate between $1$ and $-1$, with vertex $v_1$ taking labels $1$ and $-1$ with equal probability. Since the paths in path cover $\Pi'$ are pairwise vertex-disjoint, the labeling $\chi$ is well-defined over $V' := \cup_{\pi' \in \Pi'} \pi'$. We choose random labeling for all nodes in $V \setminus V'$, i.e., we independently label each node $v \in V \setminus V'$ with $\chi(v) = -1$ with probability $1/2$ and $\chi(v) = 1$ with probability $1/2$. An illustration can be found in \Cref{fig:undi-vertex-coloring}.

\begin{figure}[htp]
\centering
\includegraphics[width=\linewidth]{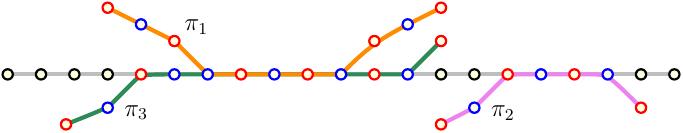}
\caption{
In this figure, paths $\pi_1, \pi_2, \pi_3 \in \Pi'$ from the path cover are intersecting a path $\pi \in \Pi$. Paths in the path cover are pairwise vertex-disjoint, and each path in the cover contributes discrepancy $0$, $-1$, or $+1$ to $\pi$.
}
\label{fig:undi-vertex-coloring}
\end{figure}

\noindent
\textbf{Bounding the discrepancy $\chi(\Pi)$.}
Fix a path $\pi \in \Pi$. We will show that $$\left| \sum_{v \in \pi} \chi(v) \right| = O(n^{1/4} \log^{1/2} (n))$$ with high probability. \Cref{thm:disc_upper_vertex} will follow as $|\Pi| = \textrm{poly}(n)$. 

\begin{proposition}    \label{prop:small_disc}

    For each path $\pi'$ in path cover $\Pi'$,
    \[
    \sum_{v \in \pi \cap \pi'} \chi(v) \in \{-1, 0, 1\}.
    \]
    If $|\pi \cap \pi'| \equiv 0 \mod 2$, then  $\sum_{v \in \pi \cap \pi'} \chi(v) = 0$. 
    Moreover,
    \[
    \Pr\left[\sum_{v \in \pi \cap \pi'} \chi(v) = -1 \right] = \Pr\left[\sum_{v \in \pi \cap \pi'} \chi(v) = 1\right].
    \]
\end{proposition}
\begin{proof}
    By the Consistency Property of $\Pi'$ (as proven in Proposition \ref{prop:path_cover}), path $\pi \cap \pi'$ is a (possibly empty) contiguous subpath of $\pi'$. Then since consecutive nodes in $\pi'$ alternate between $-1$ and $1$, it follows that $$\sum_{v \in \pi \cap \pi'} \chi(v) \in \{-1, 0, 1\}.$$
    
    Now note that $\sum_{v \in \pi \cap \pi'} \chi(v) \neq 0$ iff $|\pi \cap \pi'|$ is odd. 
    Moreover, the first vertex of $\pi \cap \pi'$ takes labels $1$ and $-1$ with equal probability. This concludes the proof of \Cref{prop:small_disc}.    
\end{proof}

We are now ready to bound the discrepancy of $\pi$.

\begin{proposition}
\label{prop:vertexdiscrepancy}
With high probability,  $ \chi(\pi)  = O(n^{1/4} \log^{1/2} n)$.
\label{prop:disc_bound}
\end{proposition}
\begin{proof}
We partition the nodes of $\pi$ into two sources of discrepancy that we will bound separately. Let $V' := \cup_{\pi' \in \Pi'} \pi'$. 

\medskip
\noindent\textbf{Discrepancy of $\pi \cap V'$.} For each path $\pi' \in \Pi'$, let $X_{\pi'}$ be the random variable defined as 
\[
X_{\pi'} := \sum_{v \in \pi \cap \pi'} \chi(v). 
\]
We can restate the discrepancy of $\pi \cap V'$ as
\[
\left| \sum_{v \in \pi \cap V'} \chi(v) \right| = \left| \sum_{\pi' \in \Pi'} X_{\pi'}\right|.
\]

By Proposition \ref{prop:small_disc}, if $|\pi \cap \pi'|$ is even, then  $X_{\pi'} = 0$. Therefore, we may assume without any loss of generality that $|\pi \cap \pi'|$ is odd for all $\pi' \in \Pi'$.  In this case,
    $
    \Pr\left[X_{\pi'} = -1 \right] = \Pr\left[X_{\pi'} = 1\right] = 1/2, 
    $ 
    implying that $\mathbb{E}[\sum_{\pi' \in \Pi'}X_{\pi'}] = 0$.
Then  by Proposition \ref{prop:path_cover} and the Chernoff bound, it follows that for any constant $c \geq 1$,
\[
\Pr\left[ \left|\sum_{\pi' \in \Pi'} X_{\pi'} \right| \geq c \cdot n^{1/4} \log^{1/2} (n) \right] \leq e^{-c^2 \frac{n^{1/2} \log(n)}{2|\Pi'|}} \leq e^{-c^2  \cdot \log (n)/2} = n^{-c^2/2}.
\]

\medskip \noindent
\textbf{Discrepancy of $\pi \setminus V'$.} Note that by the Covering Property of the path cover (as proven in Proposition \ref{prop:path_cover}), $|\pi \setminus V'| \leq n^{1/2}$. Moreover, the  nodes in $V \setminus V'$ are labeled independently at random,   implying that $\mathbb{E}\left[\sum_{v \in \pi \setminus V'} \chi(v) \right] = 0$. Then we may apply a Chernoff bound to argue that for any constant $c \geq 1$, 
\begin{align*}
\Pr\paren{ \card{\sum_{v \in \pi \setminus V'} \chi(v)}  \geq  c  \cdot n^{1/4} \log^{1/2} n} &\leq \exp\paren{-c^2 \frac{n^{1/2} \log(n)}{2|\pi \setminus V'|}} \\
&\leq \exp\paren{-c^2  \cdot \log (n)/2} \\
& = n^{-c^2/2}.     
\end{align*}
We have shown that the discrepancy of our labeling is $O(n^{1/4} \log^{1/2}(n))$ for $\pi \cap V'$ and $O(n^{1/4} \log^{1/2}(n))$ for $\pi \setminus V'$ with high probability. Therefore, with high probability, the total discrepancy of $\pi$ is $O(n^{1/4} \log^{1/2}(n))$. This completes the proof of \Cref{prop:vertexdiscrepancy}. 
\end{proof}

\medskip 
\noindent\textbf{Extending to hereditary discrepancy.}
Let $A$ be the vertex incidence matrix of a path system $S = (V, \Pi)$ on $n$ nodes, and let $A_Y$ be the submatrix of $A$ obtained by taking all of its rows but only a subset $Y$ of its columns. Then there exists a subset $V_Y \subseteq V$ of the nodes in $V$ such that $A_Y$ is the vertex incidence matrix of the path system $S[V_Y]$ (path system $S$ induced on $V_Y$). Moreover, if path system $S$ is consistent, then $S[V_Y]$ is also consistent. Then we may apply our explicit vertex discrepancy upper bound to $S[V_Y]$. We conclude that the hereditary vertex discrepancy of $S$ is $O(n^{1/4} \log^{1/2}(n))$.

\subsection{Edge Discrepancy Upper Bound -- Explicit Coloring}
\label{apdx:edge-disc-exp}
\label{subsec:edge-disc-exp}

By Theorem \ref{thm:digraph-open}, the edge discrepancy of the unique shortest paths of a (possibly directed) graph on $m$ edges is $O(m^{1/4})$. However, in the case of undirected graphs and DAGs, we can improve the edge discrepancy to $O(n^{1/4}\log^{1/2}(n))$, where $n$ is the number of vertices in the graph, by modifying the explicit construction for vertex discrepancy in \Cref{subsec:vert-disc-exp}.
Our proof strategy will follow the same framework as the explicit construction for vertex discrepancy but with some added complications in the construction and analysis.

We first introduce some new notation that will be useful in this section. Given a path $\pi$ and nodes $u, v \in \pi$, we denote by $u <_{\pi} v$ if $u$ occurs before $v$ on path $\pi$. Additionally, given a path system $S = (V, \Pi)$, we define the edge set $E \subseteq V \times V$ of the path system as the set of all pairs of nodes $u, v \in V$ that appear consecutively in some path in $\Pi$. Likewise, for any path $\pi$ over the vertex set $V$, we define the edge set of $\pi$, $E(\pi) \subseteq \pi \times \pi$, as the set of all pairs of nodes $u, v \in \pi$ such that $u, v$ appear consecutively in $\pi$ and $(u, v) \in E$. Note that if path system $S$ corresponds to paths in a graph $G$, then $E$ will be precisely the edge set of $G$. 

Recall that we wish to construct an edge labeling $\chi: E \mapsto \{-1, 1\}$ so that 
\[\chi(\Pi) = \max_{\pi \in \Pi} \left| \sum_{e \in E(\pi)} \chi(e)\right| \]
is minimized. We will upper bound the discrepancy $\chi(\Pi)$ of  consistent path systems such that $|V| = n$ and $|\Pi| = \textrm{ poly}(n)$. This immediately implies an upper bound on the edge discrepancy of unique shortest paths in undirected graphs and DAGs. 

\begin{theorem}
    For all consistent path systems $S = (V, \Pi)$ where $|V| = n$ and $|\Pi| = \textrm{\normalfont poly}(n)$ with edge set $E$,  there exists a labeling $\chi:E \to \{-1,1\}$ such that \[\chi(\Pi) = O(n^{1/4}\log^{1/2}(n)).\]
    \label{thm:disc_upper}
\sloppy    Consequently, every $n$-vertex undirected graph or directed acyclic graph has hereditary edge discrepancy $O(n^{1/4} \log^{1/2}(n))$.
\end{theorem}

Let $S = (V, \Pi)$ be a consistent path system with $|V| = n$ and $|\Pi| = \textrm{poly}(n)$. As the first step towards constructing our labeling $\chi: E \mapsto \{-1, 1\}$, we will construct a collection of paths $\Pi'$ on $V$ with a useful covering property over the paths in $\Pi$. {We give the construction for undirected graphs below, then show a standard extension to DAGs at the end of the subsection.}

\subsubsection{Constructing path cover \texorpdfstring{$\Pi'$}{Lg}} 
Initially, we let $\Pi' = \emptyset$. We define $V'$ to be the set of all nodes in $V$ belonging to a path in $\Pi'$, i.e., \[V' := \bigcup_{\pi' \in \Pi'} \pi'.\]

While there exists a path $\pi \in \Pi$ such that $\left|\pi \setminus V' \right| \geq n^{1/2}$, our goal is to find a (possibly non-contiguous) subpath of $\pi$ of length $n^{1/2}$ that is vertex-disjoint from all paths in $\Pi'$. 
Specifically, let $\pi' \subseteq \pi$ be a (possibly non-contiguous) subpath of $\pi$ containing exactly the first $n^{1/2}$ nodes in $\pi \setminus V'$.
Add path $\pi'$ to path cover $\Pi'$ and update $V'$. Repeatedly add paths to path cover $\Pi'$ in this manner until $|\pi \setminus V'| < n^{1/2}$ for all $\pi \in \Pi$. 

Note that our path cover $\Pi'$ is very similar to the path cover used in the explicit vertex discrepancy upper bound. Indeed, path cover $\Pi'$ inherits all 
properties of the path cover defined in Subsection \ref{subsec:vert-disc-exp}.
The key difference here is that we require  subpaths $\pi' \subseteq \pi$ in $\Pi'$  to contain the \textit{first} $n^{1/2}$ nodes in $\pi \setminus V'$. This will imply an additional property of our path cover, which we call the {\em No Repeats Property}.

\begin{figure}[htp]
\centering
\includegraphics[width=0.5\linewidth]{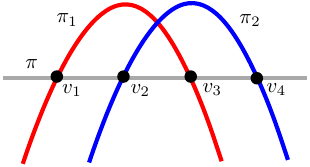}
\caption{In this figure, paths $\pi_1, \pi_2 \in \Pi'$ are intersecting a path $\pi \in \Pi$. This arrangement of paths is forbidden by the No Repeats Property of Proposition \ref{prop:edge_path_cover}. }
\label{fig:no_repeats}
\end{figure}

\begin{proposition} 
\label{prop:edge_path_cover} 
Path cover $\Pi'$ satisfies all properties of Proposition \ref{prop:path_cover}, as well as the following additional properties:
\begin{itemize}
    \item {\em (Edge Covering Property)} 
    For all $\pi \in \Pi$, the number of edges in $\pi$ that
    are not incident to any node lying in a path in path cover $\Pi'$ is at most $n^{1/2}$.  Formally, let $V' = \cup_{\pi' \in \Pi'} \pi'$. For all $\pi \in \Pi$, 
    $$
    \left| \{(u, v) \in E(\pi) \mid u \not \in V' \text{ and } v \not \in V'  \}\right| \leq n^{1/2},
    $$
     \item {\em (No Repeats Property)} 
    \sloppy For all paths $\pi \in \Pi$, $\pi_1, \pi_2 \in \Pi'$, and nodes $v_1, v_2, v_3, v_4 \in \pi$ such that $v_1, v_3 \in \pi_1$ and $v_2, v_4 \in \pi_2$, the following ordering of the vertices in $\Pi$ is impossible:
    $$
    v_1 <_{\pi} v_2  <_{\pi} v_3  <_{\pi} v_4, 
    $$
    where $x <_{\pi} y$ indicates that node $x$ occurs in $\pi$ before node $y$.
\end{itemize}
\end{proposition}
\begin{proof}
All properties from Proposition \ref{prop:path_cover} follow from an identical argument as in the original proof. The Edge Covering Property follows immediately from the Covering Property of Proposition \ref{prop:path_cover}.  What remains is to prove the No Repeats Property. 

Suppose for the sake of contradiction that there exist paths $\pi \in \Pi$, $\pi_1, \pi_2 \in \Pi'$, and nodes $v_1, v_2, v_3, v_4 \in \pi$ such that $v_1, v_3 \in \pi_1$ and $v_2, v_4 \in \pi_2$, where $ v_1 <_{\pi} v_2  <_{\pi} v_3  <_{\pi} v_4$. We will assume that path $\pi_1$ was added to $\Pi'$ before path $\pi_2$ (the case where $\pi_2$ was added to $\Pi'$ first is symmetric). By the construction of $\Pi'$, path $\pi_1 \in \Pi'$ is a (possibly non-contiguous) subpath of a path $\pi_1'' \in \Pi$ from which it is constructed. Additionally, by the consistency of the path system $S$, the intersection $\pi \cap \pi_1''$ is a contiguous subpath of $\pi$. Then $v_2 \in \pi \cap \pi_1''$, and specifically, $v_2 \in \pi_1''$.

Since $v_2 \in \pi_2$ and paths in $\Pi'$ are pairwise vertex-disjoint,  $v_2 \not \in \pi_1$. Since path $\pi_1$ was added to $\Pi'$ before path $\pi_2$, this means that when $\pi_1$ was added to $\Pi'$, node $v_2$ did not belong to any path in $\Pi'$ (i.e., $v_2$ was not in $V'$). Recall that in our construction of $\Pi'$, we constructed subpath $\pi_1 \subseteq \pi_1''$ so that it contained exactly the \textit{first} $n^{1/2}$ nodes in $\pi_1'' \setminus V'$. 
However, $v_2 \not \in \pi_1$, but $v_3 \in \pi_1$, and $v_2$ comes before $v_3$ in $\pi_1''$.
This contradicts our construction of path $\pi_1$  in path cover $\Pi'$. 

\end{proof}

\subsubsection{Constructing labeling \texorpdfstring{$\chi$}{Lg}} 
Let $\pi' \in \Pi'$ be a path of length $k$ in our path cover. Let $e_1, \dots, e_k \in E(\pi')$ be the edges in $\pi'$ listed in the order they appear in $\pi'$.
Note that since $\pi'$ is a possibly non-contiguous subpath of a path in $\Pi$,  pairs of nodes $u, v \in V$ that appear consecutively in $\pi$ do not necessarily correspond to edges in edge set $E$.

We will label the edges in $E(\pi')$ using the following random process.  With probability $1/2$ we define $\chi: E(\pi')\mapsto \{-1, 1\}$ to be 
\[
\chi(e_i) = 
\begin{cases} 
      1 &   i \equiv 0 \mod 2  \text{ and } i \in [1, k] \\
      -1 &  i \equiv 1 \mod 2  \text{ and } i \in [1, k] 
   \end{cases},
\]
and  with probability $1/2$ we define $\chi: E(\pi') \mapsto \{-1, 1\}$ to be 
\[\chi(e_i) = 
\begin{cases} 
      -1 &   i \equiv 0 \mod 2  \text{ and } i \in [1, k] \\
      1 &  i \equiv 1 \mod 2  \text{ and } i \in [1, k] 
   \end{cases}.
\]

Note that the labels of consecutive edges $e_i, e_{i+1}$ in $\pi'$ alternate between $1$ and $-1$, with edge $e_1$ taking labels $1$ and $-1$ with equal probability.

Since the paths in path cover $\Pi'$ are pairwise vertex-disjoint, the labeling $\chi$ is well-defined over 
\begin{align}
\label{eq:e'definition}
E' := \cup_{\pi' \in \Pi'} E(\pi').
\end{align} 

We take a random labeling for all edges in $E \setminus E'$, i.e., we independently label each edge $e \in E \setminus E'$ with $\chi(e) = -1$ with probability $1/2$ and $\chi(e) = 1$ with probability $1/2$.

\subsubsection{Bounding the discrepancy \texorpdfstring{$\phi$}{Lg}}
Fix a path $\pi := \pi[s, t] \in \Pi$. We will show that 
\[
\left| \sum_{e \in E(\pi)} \chi(e) \right| = O(n^{1/4} \log^{1/2} (n))
\]
 with high probability. This will complete the proof of Lemma \ref{thm:disc_upper} since $|\Pi| = \textrm{poly}(n)$. The proof of the following proposition follows from an argument identical to Proposition \ref{prop:small_disc} and hence omitted. 

\begin{proposition}
    For each path $\pi'$ in path cover $\Pi'$,
    \[\sum_{e \in E(\pi) \cap E(\pi') } \chi(e) \in \{-1, 0, 1\}.
    \]
    If $|E(\pi) \cap E(\pi')| \equiv 0 \mod 2$, then  $\sum_{e \in E(\pi) \cap E(\pi')} \chi(e) = 0$. 
Moreover,
    \[\Pr\left[\sum_{e \in E(\pi) \cap E(\pi')} \chi(e) = -1 \right] = \Pr\left[\sum_{e \in E(\pi) \cap E(\pi')} \chi(e) = 1\right].
    \]
    \label{prop:edge_small_disc}
\end{proposition}

We are now ready to bound the edge discrepancy of $\pi$. Define \[V' := \bigcup_{\pi' \in \Pi'} \pi'  \quad \text{and} \quad E' := \bigcup_{\pi' \in \Pi'} E(\pi').\]

We partition the edges of the path $\pi$ into three sources of discrepancy that we will bound separately. Specifically, using the definition of $E'$ in equation~\cref{eq:e'definition}, we  split $E(\pi) \subseteq \pi \times \pi$ into the following sets $E_1, E_2, E_3$:
\begin{itemize}
    \item $E_1 := E(\pi) \cap E'$,
    \item $E_2 := E(\pi) \cap  ((V \setminus V') \times (V \setminus V') )$, and
    \item $E_3 := E(\pi) \setminus (E_1 \cup E_2)$.
\end{itemize}
Sets $E_1$ and $E_2$ roughly correspond to the two sources of discrepancy considered in the vertex discrepancy upper bound, while set $E_3$ corresponds to a new source of discrepancy requiring new arguments to bound. We begin with set $E_1$. 
\begin{proposition}[Discrepancy of $E_1$]
With high probability,  $\left| \sum_{e \in E_1} \chi(e) \right| = O(n^{1/4} \log^{1/2} n)$.
\label{prop:edge_disc_e1}
\end{proposition}
\begin{proof}
The proposition follows from an argument similar to Proposition \ref{prop:disc_bound}. 
For each path $\pi' \in \Pi'$, let $X_{\pi'}$ be the random variable defined as 
\[X_{\pi'} := \sum_{e \in E(\pi) \cap E(\pi') } \chi(e). 
\]
We can restate the discrepancy of $E_1= E(\pi) \cap E'$ as
\[\left| \sum_{e \in E_1} \chi(e) \right| = \left|\sum_{\pi' \in \Pi'} X_{\pi'}\right|.\]

By Proposition \ref{prop:edge_small_disc},  if $|E(\pi) \cap E(\pi')| \equiv 0 \mod 2$, then  $X_{\pi'} = 0$, so without any loss of generality, we may assume that $|E(\pi) \cap E(\pi')|$ is odd for all $\pi' \in \Pi'$.
 In this case,
    $$
    \Pr\left[X_{\pi'} = -1 \right] = \Pr\left[X_{\pi'} = 1\right] = 1/2, 
    $$
    implying that $\mathbb{E}[\sum_{\pi' \in \Pi'}X_{\pi'}] = 0$.
    Then, by Proposition \ref{prop:edge_path_cover} and the Chernoff bound, it follows that for any constant $c \geq 1$,
\[\Pr\left[\left|\sum_{\pi' \in \Pi'} X_{\pi'}\right| \geq c \cdot n^{1/4} \log^{1/2}(n) \right] \leq e^{-c^2 \frac{n^{1/2} \log(n)}{2|\Pi'|}} \leq e^{-c^2  \cdot \log (n)/2} \leq n^{-c^2/2}.\]
\end{proof}
\noindent We now bound the discrepancy of $E_2 = E(\pi) \cap  ((V \setminus V') \times (V \setminus V') )$. 
\begin{proposition}[Discrepancy of $E_2$]
With high probability,  $\left| \sum_{e \in E_2} \chi(e) \right| = O(n^{1/4} \log^{1/2}(n))$.
\label{prop:edge_disc_e2}
\end{proposition}
\begin{proof}
The proposition follows from an argument similar to Proposition \ref{prop:disc_bound}. 
Note that by the Edge Covering Property of the path cover (Proposition \ref{prop:edge_path_cover}),  \[
    |E_2| = \left| \{(u, v) \in E(\pi)  \mid u, v \not \in V'  \}\right| \leq n^{1/2}.
    \]
    Moreover, the  edges in $E \setminus E'$ are labeled independently at random, so we may apply a Chernoff bound  to argue that for any constant $c \geq 1$, 
\[\Pr\left[  \left|\sum_{e \in E_2} \chi(e) \right|  \geq  c  \cdot n^{1/4} \log^{1/2}(n) \right]  \leq e^{-c^2 \frac{n^{1/2} \log(n)}{2|E_2|}} \leq e^{-c^2  \cdot \log (n)/2} \leq n^{-c^2/2}.  
\]
completing the proof.
\end{proof}

Finally, we upper bound the discrepancy of the remainder of the edges, $E_3=E(\pi) \setminus (E_1 \cup E_2)$.

\begin{proposition}[Discrepancy of $E_3$]
With high probability,  $\left| \sum_{e \in E_3} \chi(e) \right| = O(n^{1/4} \log^{1/2} (n))$.
\label{prop:edge_disc_e3}
\end{proposition}
\begin{proof}
  Take a path $\pi\in \Pi$. Let 
\[k(\pi) := \left| \{ \pi' \in \Pi' \mid \pi \cap \pi' \neq \emptyset  \} \right|
\] 
denote the number of paths in our path cover that intersect $\pi$. We define a function $f:\mathbb{Z}_{\geq 0} \mapsto \mathbb{Z}_{\geq 0}$ such that $f(\phi)$ equals the largest possible value of $|E_3|$ when $\phi = k(\pi)$, over all paths $\pi\in \Pi$. Note that $f$ is well-defined since $0 \leq |E_3| \leq |E|$. We will prove that
$f(\phi) \leq 4\phi,$ by recursively decomposing path $\pi$.

When $\phi=1$, there is only one path $\pi' \in \Pi'$ that intersects $\pi$. Then the only edges in $E_3$ are of the form
\[
E(\pi) \cap ((V' \times (V \setminus V')) \cup ((V \setminus V') \times V') ) =E(\pi) \cap ((\pi' \times (V \setminus \pi')) \cup ((V \setminus \pi') \times \pi') ).
\]
By the Consistency Property of Proposition \ref{prop:edge_path_cover}, path $\pi'$ can intersect $\pi$ and then split apart at most once. Then
\[
f(1) = |E_3| = |E(\pi) \cap ((\pi' \times (V \setminus \pi')) \cup ((V \setminus \pi') \times \pi') )| \leq 2.
\]

When $\phi > 1$, we will split our analysis into the  two cases: 
\begin{itemize}  
    \item \textbf{Case 1.} There exist paths $\pi_1', \pi_2' \in \Pi'$  and nodes $v_1, v_2, v_3 \in \pi$ such that $v_1, v_3 \in \pi_1'$ and $v_2 \in \pi_2'$ and $v_1 <_{\pi} v_2 <_{\pi} v_3$.  In this case, we can assume without any loss of generality that $\pi[v_1, v_3] \cap \pi_1' = \{v_1, v_3\}$ (e.g., by choosing $v_1, v_3$ so that this equality holds). Let $x$ be the node immediately following $v_1$ in $\pi$, and let $y$ be the node immediately preceding $v_3$ in $\pi$. Recall that $s$ is the first node of $\pi$ and $t$ is the last node of $\pi$. It will be useful for the analysis to split $\pi$ into three subpaths:
    \[    
    \pi = \pi[s, v_1] \circ \pi[x, y] \circ \pi[v_3, t],\]
    where $\circ$ denotes the concatenation operation. Define 
    \begin{align*}
        \phi_1 &:=  |\{\pi' \in \Pi' \mid \pi[x, y] \cap \pi' \neq \emptyset \}| \\
        \phi_2 &:= |\{\pi' \in \Pi' \mid (\pi[s, v_1] \circ \pi[v_3, t]) \cap \pi' \neq \emptyset \}|.
    \end{align*}

We claim that $\phi_1 < \phi$, $\phi_2 < \phi$, and $\phi_1 + \phi_2 = \phi$. We will use these facts to establish a recurrence relation for $f$. By our assumption that $\pi[v_1, v_3] \cap \pi_1' = \{v_1, v_3\}$, it follows that $\pi[x, y] \cap \pi_1' = \emptyset$, and so $\phi_1 < \phi$. Likewise, by the No Repeats Property of Proposition \ref{prop:edge_path_cover}, 
\[(\pi[s, v_1] \circ \pi[v_3, t]) \cap \pi_2' = \emptyset.\]
Therefore, $\phi_2 < \phi$. Finally, observe that more generally, if there exists a path $\pi' \in \Pi'$ such that $\pi' \cap \pi[x, y] \neq \emptyset$ and $\pi' \cap (\pi[s, v_1] \circ \pi[v_3, t]) \neq \emptyset$, then  the No Repeats Property of Proposition \ref{prop:edge_path_cover} is violated. We conclude that $\phi_1 + \phi_2 = \phi$. 

Now $|E_3|$ can be upper bounded by the following inequality:
\[|E_3| \leq |E_3 \cap E(\pi[x, y])| + |E_3 \cap E(\pi[s, v_1] \circ \pi[v_3, t])| +2.
\]
Then using the observations about $\phi_1, \phi_2$, and $\phi$ in the previous paragraph, we obtain the following recurrence for $f$:
\[f(\phi) \leq f(\phi_1) + f(\phi_2) + 2 = f(i) + f(\phi - i) + 2,
\]
where $0 < i < \phi$.

    \item \textbf{Case 2.} There exists a path $\pi' \in \Pi'$ and $v_1, v_2 \in \pi$ such that $\pi \cap \pi' = \pi[v_1, v_2] \cap V'$. Let $x$ be the node immediately preceding $v_1$ in $\pi$, and let $y$ be the node immediately following $v_2$ in $\pi$.  Again, we split $\pi$ into three subpaths:
    \[
    \pi[s, t] = \pi[s, x] \circ \pi[v_1, v_2] \circ \pi[y, t].
    \]
    Let \begin{align*}
        \phi_1 &:=  |\{\pi' \in \Pi' \mid \pi[v_1, v_2] \cap \pi' \neq \emptyset \}| \\ 
        \phi_2 &:= |\{\pi' \in \Pi' \mid (\pi[s, x] \circ \pi[y, t]) \cap \pi' \neq \emptyset \}|.
    \end{align*}
    Our assumption in Case 2 follows that $\phi_1 = 1$ and $\phi_2 = \phi - 1$.
Since $|E_3|$ can be upper bounded by the inequality
\[|E_3| \leq |E_3 \cap E(\pi[v_1, v_2])| + |E_3 \cap E(\pi[s, x] \circ \pi[y, t])| +2,
\] we immediately obtain the recurrence
\[    f(\phi) \leq f(\phi_1) + f(\phi_2) + 2 \leq f(1) + f(\phi - 1) + 2.
\]
\end{itemize}  
Taking our results from Case 1 and Case 2 together, we obtain the recurrence relation
\[f(\phi) \leq \begin{cases}
    \max\left\{f(i) + f(\phi - i) + 2,\text{ } f(1) + f(\phi - 1) + 2\right\} & \phi > 1 \text{ and } 1 < i < \phi \\
    2 & \phi = 1
\end{cases}.
\]
Applying this recurrence at most $\phi$ times, we find that
\[f(\phi) \leq \phi\cdot f(1) + 2\phi \leq 4\phi.
\]
Finally, since $k \leq |\Pi'| \leq n^{1/2}$ and we defined $f$ so that $f(k)$ equals the largest possible value of $|E_3|$, we conclude that 
\[|E_3| \leq f(k) \leq f(n^{1/2}) = O(n^{1/2}).
\]
Since the edges in $E_3 \subseteq E \setminus E'$ are labeled independently at random, we may apply a Chernoff bound as in Proposition \ref{prop:edge_disc_e2} to argue that $\chi(E_3) = O(n^{1/4} \log^{1/2}(n))$ with high probability.
\end{proof}

We have shown that with high probability, the discrepancy of our edge labeling is $O(n^{1/4} \log^{1/2}(n))$ for $E_1$, $E_2$, and $E_3$,
so we conclude that the total discrepancy of $\pi$ is $O(n^{1/4} \log^{1/2}(n))$. 

{
\paragraph{Extending to hereditary discrepancy}

Let $A$ be the edge incidence matrix of a path system $S = (V, \Pi)$ on $n$ nodes, and let $A_Y$ be the submatrix of $A$ obtained by taking all of its rows but only a subset $Y$ of its columns. We can rephrase the problem of bounding the discrepancy of $A_Y$ as the following edge coloring problem.

Let $E$ be the edge set of path system $S$, and let $E_Y \subseteq E$ be the set of edges in $E$ associated with the set of columns $Y$ in matrix $A$. Then bounding the discrepancy of $A_Y$ is equivalent to constructing an edge labeling $\chi:E_Y \mapsto \{-1, 1\}$ so that $\chi(E_Y) := \max_{\pi \in \Pi}\left| \sum_{e \in E(\pi) \cap E_Y} \chi(e) \right|$ is minimized. 

In this setting, for each $\pi \in \Pi$, we redefine $E(\pi) \subseteq \pi \times \pi$ to be the set of nodes $u, v \in \pi$ such that $u, v$ appear consecutively in $\pi$ and $(u, v) \in E_Y$. With this new definition of $E(\pi)$, we can construct an edge labeling $\chi$ using the exact same procedure described in this section. Indeed, we claim that this will yield the same discrepancy bound of $\chi(E_Y) = O(n^{1/4} \log^{1/2} n)$. 

The key observation needed to confirm this claim is that the consistency of path system $S$ extends to the setting where $S$ is restricted to edge set $E_Y$. Formally, for all paths $\pi_1, \pi_2 \in \Pi$, the set $E(\pi_1) \cap E(\pi_2)$ is a (possibly empty) contiguous subsequence of the edge sets $E(\pi_1)$ and $E(\pi_2)$. Using this observation in place of the notion of consistency, the discrepancy bounds on edge sets $E_1, E_2, E_3 \subseteq E(\pi)$ follow from arguments identical to those of Propositions \ref{prop:edge_disc_e1}, \ref{prop:edge_disc_e2}, and \ref{prop:edge_disc_e3}. We conclude that the hereditary edge discrepancy of $S$ is $O(n^{1/4} \log ^{1/2} n)$. 
}

{
\paragraph{Extending to DAGs}
The construction above for undirected graphs relies on two properties mentioned in \Cref{prop:edge_path_cover}, the \emph{Edge Covering Property} and \emph{No Repeats Property}. The former follows from the construction of the path cover alone and is unaffected by edge orientations. The latter is the only place where consistency enters: if two nodes lie on two paths, then every node between them on one path also lies on the other. In a directed graph, consistency guarantees this only when the two paths visit the two nodes in the same order. In a DAG they always do, the order of two nodes on a path follows the topological order and is the same on every path containing both. The same reasoning applies to the Consistency Property used in \Cref{prop:path_cover}, therefore our vertex and edge discrepancy upper bounds for undirected graphs hold for DAGs as well. 
}

\section{Planar Graphs}
\label{sec:planar}

In this section, we will extend our hereditary vertex discrepancy lower bound for unique shortest paths in undirected graphs to the planar graph setting. 
\begin{theorem}
There exists an $n$-vertex undirected planar graph with hereditary vertex discrepancy at least $\Omega\left({n^{1/4}\over \log^{2}(n)}\right)$.
\label{thm:planar_lb}
\end{theorem}
To prove this theorem, we will first give an abbreviated presentation of the graph construction in \cite{bodwin2023folklore} that we used implicitly to obtain the  $\Omega(n^{1/4}/\sqrt{\log n})$ hereditary vertex discrepancy lower bound in Theorem \ref{thm:undi-hdisc-lb}. Then we will describe a simple procedure to make this graph planar and argue that the shortest path structure of this planarized graph remains unchanged.

\subsection{Graph Construction of Bodwin and Hoppenworth}
\label{sec:bodwin2023folkloregraph}

Take $n$ to be a large enough positive integer, and take $p =n \log n$. We will describe an $n$-node  weighted undirected graph $G = (V, E, w)$ originally constructed in Bodwin and Hoppenworth \cite{bodwin2023folklore}.

\paragraph{Vertex Set $V$}

We will use $\ell = \Theta\left( \frac{n^{1/2}}{ \log n} \right)$ as a positive integer parameter for our construction. The graph $G$ we create will consist of $\ell$ layers, denoted as $L_1, \dots, L_{\ell}$. Each layer will have $n/\ell$ nodes, arranged from $1$ to $n/\ell$. Initially, we will assign a tuple label $(i, j)$ to the $j$th node in the $L_i$ layer. We will interpret the node labeled $(i, j)$ as a point in $\mathbb{R}^2$ with integral coordinates. The vertex set $V$ of graph $G$ is made up of these $n$ nodes distributed across $\ell$ layers.

    Next we will randomize the node labels in $V$. For each layer $L_i$, where $i$ ranges from 1 to $\ell$, we randomly and uniformly pick a real number in the interval $(0, 1)$  and we call it $\psi_i$. After that, for each node in layer $L_i$ of the graph $G$ that is currently labeled $(i, j)$, we relabel it as 
    \[      
    \left(i, j+\sum_{k=1}^j \psi_k\right).
    \]
    These new labels for the nodes in $V$ are also treated as points in $\mathbb{R}^2$. We can imagine this process as adding a small epsilon of structured noise to the points corresponding to the nodes in the graph. The purpose of this noise is technical, but serves the purpose of achieving `symmetry breaking' (see Section 2.4 of \cite{bodwin2023folklore} for details).

\paragraph{Edge Set $E$} 
All edges will be between subsequent layers $L_i, L_{i+1}$ within $G$. It will be helpful to think of the edges in $G$ as directed from $L_i$ to $L_{i+1}$, although in actuality $G$ will be undirected. 
We represent the set of edges in $G$ between layers $L_i$ and $L_{i+1}$ as $E_i$. For any edge $e = (v_1, v_2) \in E$, the edge $e$ will be associated with the specific vector $\Vec{u}_e := v_2 - v_1$. The 2nd coordinate of $\Vec{u}_e$ will be labeled as $u_e$. Hence, for all $e$ found in $E$, $\Vec{u}_e$ is written as  $(1, u_e)$.

    For each $i \in [1, \ell - 1]$, let 
    \[C_i := \{ (1, \psi_{i+1} + x): {x\in \mathbb{Z}, x \in [0, n/\ell^2]} \}.\]
      We will refer to the vectors in $C_i$ as \emph{edge vectors},  and note that $C_{i}$ has $O(\log^2 n) $ vectors due to the choice of $\ell$.  
    For each $v \in L_i$ and edge vector $\Vec{c} \in C_i$, if $v + \Vec{c} \in V$, then add edge $(v, v+ \Vec{c})$ to $E_i$.  After adding these edges to $E_i$, we will have that
    \[C_i = \{\Vec{u}_e \mid e \in E_i\}.
    \]

Finally, for each $e \in E$, if $\Vec{u}_e = (1, u_e)$, then we assign edge $e$ the weight $w(e) := u_e^2$. 
This completes the construction of our graph $G = (V, E, w)$.

\begin{proposition}
\label{prop:planar_draw}
Consider the graph drawing of graph $G$ where the nodes $v$ in $V$ are drawn as points at their associated coordinates in $\mathbb{R}^2$ and the edges $(u, v)$ in $E$ are drawn as straight-line segments from $u$ to $v$. This graph drawing has $O(n \log^6 n)$ edge crossings. 
\end{proposition}
\begin{proof}
First note that if edges $e_1, e_2 \in E$ cross in our graph drawing of $G$, then edges $e_1$ and $e_2$ are between the same two layers of $G$  (i.e., $e_1, e_2 \in L_i \times L_{i+1}$ for some $i \in [1, \ell - 1]$). Additionally, all edges between $L_i$ and $L_{i+1}$ are from the $j$th vertex in $L_i$ to the $(j+k)$th vertex in $L_{i+1}$, where $j \in [1, n/\ell]$ and $k \in [0, \Theta(\log^2 n)]$.

Now fix an edge $(u, v) \in E \cap (L_i \times L_{i+1})$ for some $i \in [1, \ell - 1]$. If an edge $(u', v') \in E \cap (L_i \times L_{i+1})$ crosses  $(u, v)$, then $|u - u'| \leq \log^2n$. Then there are at most $O(\log^2 n)$ nodes incident to edges that cross $(u, v)$ in our drawing.  
Since each node in $G$ has degree $O(\log^2n)$, this implies that at most $O(\log^4n)$ edges cross $(u, v)$ in our drawing. Since $|E| = O(n\log^2 n)$, we conclude that our graph drawing has $O(n \log^6 n)$ edge crossings. 
\end{proof}

\paragraph{Direction Vectors  and  Paths $\Pi$}
\label{subsec:critical_paths}
Our next step is to generate a set of unique shortest paths $\Pi$.
The paths $\Pi$ are identified by first constructing a set of vectors $D \subseteq \mathbb{R}^2$ called \textit{direction vectors}, which are defined next.

 Let $q = \Theta\left(\frac{\ell}{\log n}\right)=\Theta\left(\frac{n^{1/2}}{\log^2 n}\right)$ be an integer.  We choose our set of direction vectors $D$ to be 
    \[
    D := \left\{ \left(1, \text{ } x + \frac{y}{q} \right) \quad \text{such that} \quad x \in \left[ 1,  \frac{n}{4\ell^2} - 1 \right] \text{ and } y \in [0, q]  \right\}. 
    \]
    Note that adjacent direction vectors in $D$ differ only by $1/q$ in their second coordinate. Each of our  paths $\pi$ in $\Pi$ will have an associated direction vector $\Vec{d} \in D$, and for all $i \in [1, \ell - 1]$, path $\pi$ will take an edge vector in $C_i$ that is closest to $\Vec{d}$ in some sense. 

\paragraph{ Paths $\Pi$}
We first define a set $S \subseteq L_1$ containing half of the nodes in the first layer $L_1$ of $G$:
\[S := \left\{ (1, j + \psi_1) \in L_1  \quad \text{such that} \quad j \in \left[1, \frac{n}{2 \ell}\right] \right\}.
\]
We will define a set of pairs of nodes $P$ so that $P \subseteq S \times L_{\ell}$.
For every node $s \in S$ and direction vector $\Vec{d} \in D$, we will identify a pair of endpoints $(s, t) \in S \times L_{\ell}$ and a corresponding unique shortest path $\pi_{s, t}$ to add to $\Pi$.

Let $v_1 \in S$, and let $\Vec{d} = (1, d) \in D$.
The associated path $\pi$ has a start node $v_1$.
We iteratively grow $\pi$, layer-by-layer, as follows.
Suppose that currently $\pi = (v_1, \dots, v_i)$, for $i<\ell$, with each $v_i \in L_i$.
To determine the next node $v_{i+1} \in L_{i+1}$, let $E_i^{v_i} \subseteq E_i$ be the edges in $E_i$ incident to $v_i$, and let
\[e_i := \text{argmin}_{e \in E_i^{v_i}}(|u_e - d|).\]
By definition, $e_i$ is an edge whose first node is $v_i$; we define $v_{i+1} \in L_{i+1}$ to be the other node in $e_i$, and we append $v_{i+1}$ to $\pi$.  After this process terminates, we will have a path 
$
\pi = (v_1, \dots, v_{\ell}).
$
Denote $\pi$ as $\pi_{v_1, v_{\ell}}$ and add path $\pi_{v_1, v_{\ell}}$ to $\Pi$. Repeating for all $v_1 \in S$ and $\vec{d} \in D$ completes our construction of $\Pi$. Note that although we did not prove it, each path $\pi_{s, t} \in \Pi$ is a unique shortest $s \leadsto t$ path in $G$ by Lemma 2 of Bodwin and Hoppenworth \cite{bodwin2023folklore}.

Lemma 1 of Bodwin and Hoppenworth \cite{bodwin2023folklore} summarizes the key properties of $G, \Pi$ that are needed to prove the hereditary vertex discrepancy lower bound for unique shortest paths in undirected graphs in Theorem \ref{thm:undi-hdisc-lb}. We restate this key lemma in \Cref{lem:hopset_lemma} of  \Cref{sec:genundir}.

To obtain a $\tilde{\Omega}(n^{1/4})$ lower bound for hereditary vertex discrepancy of unique shortest paths in \textit{planar} graphs, we need to convert the graph $G$ into a planar graph while ensuring that the unique shortest path structure of the graph remains unchanged.

\subsection{Planarization of Graph \texorpdfstring{$G$}{Lg}}

In the previous subsection, we outlined the construction of the graph $G = (V, E, w)$ and set of paths $\Pi$ from \cite{bodwin2023folklore}. This graph has an associated graph drawing with $\tilde{O}(n)$ edge crossings, by Proposition \ref{prop:planar_draw}. We will now `planarize' graph $G$ by embedding it within a larger planar graph $G'$. We will use the standard strategy of replacing each edge crossing in our graph drawing of $G$ with a new vertex, thereby subdividing each crossed edge into a path. 

\paragraph{Planarization Procedure}
\begin{enumerate}
    \item  We start with the current non-planar graph $G = (V, E, w)$ with the associated graph drawing described in Proposition \ref{prop:planar_draw}.
    \item  For every edge crossing in the drawing of $G$, letting point $p \in \mathbb{R}^2$ be the location of the crossing, draw a vertical line in the plane through   $p$.  Add a new node to graph $G$ at every point where this vertical line intersects the drawing of an edge. This step may blow up the number of nodes in the graph by quite a lot, but the resulting graph will be planar, and additionally it will be layered. Denote the planar graph resulting from this procedure as $G' = (V', E', w')$. 
    
    \item We set the edge weights $w'$ of $G'$ as follows.  Each edge $e \in E'$ has an edge vector of the form $\alpha \cdot \vec{c}$, where $\alpha \in (0, 1)$ and $\vec{c} \in C_i$ for some $i \in [1, \ell - 1]$. If $\vec{c} = (1, c)$, then we assign edge $e$ weight $w'(e) = \alpha c^2$.

    \item Finally, we remove excess nodes added to the graph in step 2. While there exists a node $v$ of degree $2$ in the resulting graph, we perform the following operation. Let $(x, v)$ and $(v, y)$ be the two edges incident to $v$. Add edge $(x, y)$ to the graph and assign it weight $w'((x, v)) + w'((v, y))$. Remove node $v$ and edges $(x, v)$ and $(v, y)$ from the graph. Note that the graph will remain planar after this operation.
\end{enumerate}
Denote the planar graph resulting from this procedure as $G'' = (V'', E'', w'')$. 

\begin{proposition}
    Graph $G''$ is planar and has $O(n \log^6 n)$ nodes.
\label{prop:planar_small_size}
\end{proposition}
\begin{proof}
    Follows immediately from Proposition \ref{prop:planar_draw} and the planarization procedure.  
\end{proof}

{
\begin{remark}
    Graph $G''$ is obtained from graph $G'$ by deleting degree-two nodes and updating the edge weights appropriately. This is necessary to bound the number of vertices in $G''$ in \Cref{prop:planar_small_size}, but it does not affect the shortest path distances of the graph. For the remainder of this section, we will work directly with graph $G'$ for convenience.  
\end{remark}
}

\paragraph{Unique Shortest Paths in $G'$}
Each edge $e = (u, v) \in E$ in graph $G$ is the preimage of a $u \leadsto v$ path $\pi_e$ in graph $G'$ resulting from our planarization procedure. Likewise, each  path $\pi \in \Pi$ is the preimage of a path $\pi'$ in $G'$ obtained by replacing each edge $e \in \pi$ with path $\pi_e$.
Let the set $\Pi'$ of paths in $G'$ denote the image of the set of paths $\Pi$ in $G$ under our planarization procedure. As a final step towards proving Theorem \ref{thm:planar_lb}, we need to argue that the unique shortest path structure of $G$ is unchanged by our planarization procedure.

\begin{lemma}
Each path in $\Pi'$ is the unique shortest path between its endpoints in $G'$.
    \label{lem:planar_usp}
\end{lemma}

We now verify that graph $G'$ and paths $\Pi'$ have the unique shortest path property as stated in \Cref{lem:planar_usp}. We will first need to establish several properties of graph $G'$ and paths $\Pi'$. First, we formally state the ``layered'' property of  $G'$ and $\Pi'$. 

{
\begin{claim}
    \label{clm:layered} The vertices of graph $G'$ can be partitioned into disjoint layers $L'_1, \dots, L'_{\ell'}$ so that each path $\pi \in \Pi'$ starts in the first layer $L'_1$, ends in the last layer $L'_{\ell'}$, and contains exactly one node in each layer. 
\end{claim}
\begin{proof}
    Paths in $\Pi$ are layered in $G$, with each path starting in layer $L_1$ and ending in layer $L_{\ell}$. 
    Step 2 of the planarization procedure ensures that paths in $\Pi'$ remain layered in $G'$. This can be seen by taking each vertical line used to subdivide edges in $G'$ to be a new layer of the graph. 
\end{proof}
}

{
Recall that $L_1', \dots, L_{\ell'}'$ are the layers of graph $G'$ identified in \Cref{clm:layered}. We now prove the following claim about shortest paths in $G'$.

\begin{claim}
    \label{clm:layer-respecting-paths}
    Fix an $(s, t)$-path $\pi \in \Pi'$ in $G'$. We may assume without loss of generality that any shortest $(s, t)$-path in $G'$ uses exactly one edge in $L_i' \times L_{i+1}'$ for each $i \in [1, \ell'-1]$. 
\end{claim}
\begin{proof}
    Fix an $(s, t)$-path $\pi \in \Pi'$ in $G'$.  Since graph $G'$ is layered, any $(s, t)$-path in $G'$ must use at least one edge in $L_i' \times L_{i+1}'$ for each $i \in [1, \ell'-1]$. We can ensure that any shortest $(s, t)$-path in $G'$ uses exactly  one edge in $L_i' \times L_{i+1}'$ for each $i \in [1, \ell'-1]$ using the following standard reduction. Let $W := 1 + \sum_{e \in E'}w'(e)$.  For each edge $e \in E'$, update the weight of edge $e$ to be  $w'(e) \gets w'(e) + W$. Then any shortest $(s, t)$-path $\pi'$ in $G'$ must satisfy $|\pi'| \le |\pi|$, due to the $+W$ terms added to the edge weights of $G'$. We conclude that any shortest $(s, t)$-path in $G'$ uses exactly one edge in $L_i' \times L_{i+1}'$ for each $i \in [1, \ell'-1]$. 
    This reduction also appears in Section 4.6 of \cite{bodwin2023folklore}. 
\end{proof}
}

{
We will also require the following proposition about the construction of graph $G$ from \cite{bodwin2023folklore} that we state without proof.
\begin{proposition}[c.f. Proposition 1 of \cite{bodwin2023folklore}]
With probability $1$, for every $i \in [1, \ell -1]$ and every direction vector $\Vec{d} = (1, d) \in D$, there is a unique vector $(1, c) \in C_i$ that minimizes $|c - d|$ over all choices of $(1, c) \in C_i$. 
\label{prop:distinct_obj_dir}
\end{proposition}
}

\noindent
Additionally, our unique shortest paths argument will make use of the following technical proposition that was proven in a slightly weaker form in~\cite{bodwin2023folklore}. We provide a full proof of this proposition for completeness. 

{
\begin{proposition}[c.f. Proposition 3 of \cite{bodwin2023folklore}]
\label{prop:squared_prop}
    Let $\alpha_1, \dots, \alpha_k, b, x_1,  \dots, x_k \in \mathbb{R}^+$. Now consider $\hat{x}_1,  \dots, \hat{x}_k$ such that 
    \begin{itemize}
        \item $|\hat{x}_i - b| \leq |x_i - b|$ for all $i \in [1, k]$, and
        \item $\sum_{i=1}^k \alpha_i x_i = \sum_{i=1}^k  \alpha_i \hat{x}_i$. 
    \end{itemize}
    Then 
    \[
    \sum_{i=1}^k  \alpha_i x_i^2 \geq \sum_{i=1}^k \alpha_i \hat{x}_i^2,
    \]
    with equality only if  $|\hat{x}_i - b| = |x_i - b|$ for all $i \in [1, k]$.
\end{proposition}
\begin{proof}
Use the identity
\[
t^2=(t-b)^2+2b(t-b)+b^2 \qquad (t\in\mathbb{R}).
\]
Applying this to $t=x_i$ and $t=\hat x_i$, multiplying by $\alpha_i$, and summing over $i$ gives
\begin{align*}
\sum_{i=1}^k \alpha_i x_i^2-\sum_{i=1}^k \alpha_i \hat x_i^2
&=\sum_{i=1}^k \alpha_i\Bigl((x_i-b)^2-(\hat x_i-b)^2\Bigr)
+2b\sum_{i=1}^k \alpha_i(x_i-\hat x_i).
\end{align*}
By assumption, $\sum_{i=1}^k \alpha_i x_i=\sum_{i=1}^k \alpha_i \hat x_i$, hence
$\sum_{i=1}^k \alpha_i(x_i-\hat x_i)=0$, so
\[
\sum_{i=1}^k \alpha_i x_i^2-\sum_{i=1}^k \alpha_i \hat x_i^2
=\sum_{i=1}^k \alpha_i\Bigl((x_i-b)^2-(\hat x_i-b)^2\Bigr).
\]
Since $|\hat x_i-b|\le |x_i-b|$, we have $(\hat x_i-b)^2\le (x_i-b)^2$ for each $i$, so every summand on the right is nonnegative. Therefore
\[
\sum_{i=1}^k \alpha_i x_i^2 \ge \sum_{i=1}^k \alpha_i \hat x_i^2.
\]
If equality holds, then
\[
0=\sum_{i=1}^k \alpha_i\Bigl((x_i-b)^2-(\hat x_i-b)^2\Bigr)
\]
is a sum of nonnegative terms with $\alpha_i>0$, hence each term must be zero:
\[
(x_i-b)^2=(\hat x_i-b)^2 \quad \text{for all } i.
\]
Equivalently, $|x_i-b|=|\hat x_i-b|$ for all $i$. The converse is immediate.
\end{proof}
}

Using Claims \ref{clm:layered} and \ref{clm:layer-respecting-paths}, and Propositions \ref{prop:distinct_obj_dir} and \ref{prop:squared_prop}, we can now prove Lemma \ref{lem:planar_usp}.

\begin{proof}[Proof of Lemma \ref{lem:planar_usp}]

Fix an $s \leadsto t$ path $\pi' \in \Pi'$ in graph $G'$, and let path $\pi \in \Pi$ in $G$ be the associated preimage of $\pi'$. Let $(1, x) \in D$ be the direction vector associated with path $\pi$. Note that by Proposition \ref{prop:distinct_obj_dir}, for each layer $L_i$, there is a unique vector $(1, c) \in C_i$ that minimizes $|c-x|$ over all choices of $(1, c) \in C_i$. By our construction of the paths in $\Pi$, path $\pi$ will travel along an edge with edge vector $(1, c)$. 

In graph $G'$, there are additional layers between layers $L_i$ and $L_{i+1}$, due to step 2 of our planarization procedure (as proved in \Cref{clm:layered}). If path $\pi$ traveled along an edge with edge vector $(1, c)$ from $L_i$ to $L_{i+1}$ in $G$, then in each layer $L_j'$ in $G'$ between $L_i$ and $L_{i+1}$, path $\pi'$ will take an edge with corresponding edge vector $\alpha \cdot (1, c)$, where $\alpha \in (0, 1)$. Moreover,  by Proposition \ref{prop:distinct_obj_dir}, this edge vector $\alpha \cdot (1, c)$ will be the unique edge vector from layer $L_j'$ minimizing $|c - x|$.

Recall that $\ell'$ is the number of layers in $G'$. 
Let $\hat{x}_1, \dots, \hat{x}_{\ell' - 1} \in \mathbb{R}^+$ be  real numbers such that the $i$th edge of $\pi'$ has the corresponding vector $\alpha_i \cdot (1, \hat{x}_i)$ for $i \in [1, \ell' - 1]$ and $0 < \alpha_i \leq 1$. Then the $i$th edge of $\pi'$ has weight $\alpha_i \hat{x}_i^2$ in $G'$. 
Now consider an arbitrary shortest $s \leadsto t$ path $\pi^*$ in $G'$ such that $\pi^* \ne \pi'$. By \Cref{clm:layer-respecting-paths}, we may assume that $\pi^*$ has exactly $\ell' - 1$ edges. Let $x_1, \dots, x_{\ell' - 1} \in \mathbb{R}^+$ be real numbers such that the $i$th edge of $\pi^*$ has the corresponding vector $\alpha_i \cdot (1, x_i)$ for $i \in [1, \ell' - 1]$ and $0 < \alpha_i \leq 1$.
Then the $i$th edge of $\pi^*$ has weight $\alpha_i x_i^2$ in $G'$. 
Now observe that since $\pi^*$ and $\pi'$ are both $s \leadsto t$ paths, it follows that \[\sum_{i=1}^{\ell' - 1} \alpha_i \hat{x}_i = \sum_{i=1}^{\ell' - 1} \alpha_i x_i.\] Additionally, by our construction of $\pi'$, it follows that  \[
|\hat{x}_i - x | \leq |x_i - x|
\]
for all $i \in [1, \ell' - 1]$. In particular, since $\pi^* \neq \pi'$, there must be some $j \in [1, \ell' - 1]$ such that $\hat{x}_j \neq x_j$, and so by Proposition \ref{prop:distinct_obj_dir},  $|\hat{x}_j - x | < |x_j - x|$ with probability 1.  
Then by Proposition \ref{prop:squared_prop},
\[w(\pi') = \sum_{e \in \pi'}w(e) = \sum_{i=1}^{\ell' - 1} \alpha_i \hat{x}_i^2 < \sum_{i=1}^{\ell' - 1} \alpha_i x_i^2 = \sum_{e \in \pi^*} w(e) = w(\pi^*).
\]
This implies that the path $\pi'$ is a unique shortest $s \leadsto t$ path in $G'$, as desired.  
\end{proof}

\paragraph{Finishing the Proof}

\begin{lemma}[c.f. Lemma 1 of \cite{bodwin2023folklore}]
\label{lem:planar_construction_lemma}
There is an infinite family of $\Theta(n \log^6 n)$-node planar undirected weighted graphs $G' = (V', E', w')$ and sets $\Pi'$ of $|\Pi'|= n \log n$ paths in $G'$ with the following properties:
\begin{itemize}

\item Each path in $\Pi'$ is the unique shortest path between its endpoints in $G$.

\item Let $G$ be the $n$-node undirected weighted graph and let $\Pi$ be the set of $|\Pi| = n \log n$ paths described in \Cref{lem:hopset_lemma} when $p = n \log n$. Then $\Pi$ is an induced path subsystem of $\Pi'$. 
\end{itemize}
\label{lem:final_planar}
\end{lemma}
\begin{proof}
This follows immediately from Proposition \ref{prop:planar_small_size}, Lemma \ref{lem:planar_usp}, and the above discussion about the set of paths $\Pi'$ in $G'$. 
\end{proof}

Let $N := \Theta(n \log^6 n)$ be the number of nodes in $G'$. By Lemma \ref{lem:final_planar}, 
$
\hdiscv(\Pi') \geq \hdiscv(\Pi).
$
Likewise, by the proof of Theorem \ref{thm:undi-hdisc-lb}, $\hdiscv(\Pi) \geq \Omega(n^{1/4})$. We conclude that
\[\hdiscv(\Pi') \geq \hdiscv(\Pi)\geq \Omega(n^{1/4}/\sqrt{\log n}) = \Omega\left({N^{1/4} \over \log^{2}N}\right).\]
This completes the proof of \Cref{thm:planar_lb}.

\section{Trees and Bipartite Graphs}
\label{apdx:tree-bipartite}

For graphs with simple topology such as line, tree and bipartite graphs, both of the vertex and edge discrepancy are constant. However, a distinction can be observed on hereditary discrepancy for bipartite graphs. Formally, we have the following results.

\begin{lemma}
    \label{lem:bipartite-graph}
    Any bipartite graph has $\Theta(1)$ discrepancy. But there exist $n$-vertex weighted bipartite graphs with $\tilde{\Omega}(n^{1/4})$ hereditary discrepancy.
\end{lemma}

\begin{proof}
    To start with, it is obvious that a lower bound of $\Omega(1)$ on both edge and vertex (hereditary) discrepancy always holds for any family of graphs. We therefore first focus on the $O(1)$ discrepancy upper bound for bipartite graphs (including trees).
    
\paragraph{Analysis of discrepancy} We start with the vertex discrepancy. For a bipartite graph $G = (L \cup R, E)$, a simple scheme achieves constant vertex discrepancy: assign coloring `$+1$' to every $v \in L$ and `$-1$' to every $u \in R$. Observe that every shortest path either has a length of 1 or alternates between $L$ and $R$, thus summing up assigned colors along the shortest path gives vertex discrepancy at most $1$. Finally, we apply \Cref{obs:vertex-implies-edge} to argue that the edge discrepancy is also $O(1)$.
    
\paragraph{Analysis of hereditary discrepancy}
Observe that the graph construction in \Cref{thm:undi-hdisc-lb} is a layered graph; therefore, it is bipartite, and the hereditary discrepancy follows.

\end{proof}

\begin{lemma}
\label{lem:tree-graph}
Let $T=(V,E,w)$ be an undirected tree graph. The hereditary (vertex or edge) discrepancy of the shortest path system induced by $T$ is $\Theta(1)$.
\end{lemma}
\begin{proof}

By \Cref{lem:bipartite-graph}, a tree graph has $\Theta(1)$ discrepancy. Furthermore, the hereditary discrepancy is also $\Theta(1)$. For any vertex set $S\subseteq V$, we color them in a top-down manner. The first vertex in $S$ along a path from the root to any leaf is colored $+1$, and a vertex in $S$ is colored the opposite label from the closest ancestor in $S$. This will give vertex hereditary discrepancy $\Theta(1)$. A similar argument for edge coloring shows that edge-hereditary discrepancy is also $\Theta(1)$.
\end{proof}

\section{Bounds on $\ell_2$-discrepancy}\label{sec:L2}

Consider a set of $p$ paths on a graph $G$, the incidence matrix $A$ of $p$ rows and $n$ columns has the $(i, j)$-th element to be $1$ if the corresponding vertex $v_j$ stays on the $i$-th path, and $0$ otherwise. Then we consider the vertex $\ell_2$-discrepancy $\disc_2(A)$ and hereditary discrepancy $\herdisc_2(A)$.

Since $\disc_2(A)\leq \disc(A)$ and $\herdisc_2(A)\leq \herdisc(A)$, the upper bounds on discrepancy and hereditary discrepancy for path systems remain as upper bounds for $\ell_2$ (hereditary) discrepancy, that is, $\tilde{O}(\sqrt{n})$ for general paths (that are not necessary shortest paths and not necessarily consistent), and $\tilde{O}(n^{1/4})$ for shortest paths or consistent paths.

For lower bound on $\ell_2$ discrepancy, we first recall the following result in Larsen~\cite{larsen2017constructive}:
\begin{lemma}[\cite{larsen2017constructive}]
\label{lem:larsen2017ell2}
    For an $m \times n$ real matrix $A$, let $\lambda_1 \geq \lambda_2 \geq \cdots \lambda_n \geq 0$ denote the eigenvalues of $A^\top A$. For all positive integers $k \leq  \min\{n, m\}$, we have
    \[
    \herdisc_2(A) \geq {k\over e} \sqrt{\lambda_k \over 8 \pi m n}.
    \]
\end{lemma}
In the proof of the trace bound, Chazelle and Lvov~\cite{Chazelle2000-ld} have shown that for $k= \tr^2(M)/(2 \tr M^2)$, we have $\lambda_k \geq \tr(M)/4\min\{n,m\}$. Setting it in \Cref{lem:larsen2017ell2}, we have the same asymptotic trace bound for $\ell_2$-hereditary as in \Cref{lem:larsen}.
\[
\herdisc_2(A) \geq {\tr^2(M) \over 2e \min\{m,n\} \tr M^2} \sqrt{\tr(M) \over 32 \pi \max\{m,n\}}.
\]

Using the same calculation as in the proof of \Cref{thm:undi-hdisc-lb}, we get the equivalent  bound on $\ell_2$-hereditary discrepancy for the set $\Pi$ of shortest paths in \Cref{lem:hopset_lemma}:
\[
\herdisc_{v,2}(\Pi) \geq \Omega \paren{n^{1/4} \over \sqrt{\log(n)}}.
\]
Notice that the same lower bound argument generates a lower bound of $\Omega(n^{1/6})$ on $\ell_2$ discrepancy for the Erd\H{o}s point-line system.

Again if we drop the consistency property (or uniqueness of shortest paths), the $\ell_2$ hereditary discrepancy can be much higher. Consider a graph $G$ with two vertices $s$ and $t$, together with $2n$ vertices $u_i, v_i$, $i\in [n]$. We connect $s$ with $u_1, v_1$ and $t$ with $u_n, v_n$. In addition, $u_i, v_i$ are both connected to $u_{i+1}, v_{i+1}$, for $1\leq i\leq n-1$. Now we can encode the $n\times n$ Hadamard matrix $H$ by $n$ paths from $s$ to $t$. Each row of $H$ corresponds to a path $P_i$. If the $j$th element is $1$, we take $u_j$, otherwise, we take $v_j$. The incidence matrix $A$ considering only vertices $v_1, v_2, \cdots v_n$ would be precisely $\frac{1}{2}(H+J)$, where $J$ is an $n\times n$ matrix of all $1$. It is known that $\|Ax\|_2^2\geq n(n-1)/4$~\cite{Chazelle2000-rd}. Thus $\herdisc_2(A)=\Omega(\sqrt{n})$.

In summary, all bounds of hereditary discrepancy presented in the paper hold for $\ell_2$ hereditary discrepancy for path systems on a graph.

\section{Applications to Differential Privacy}
\label{sec:app}

In light of our new unique shortest path hereditary discrepancy lower bound result,
significant progress can be made towards closing the gap in the error bounds for the problem of Differentially Private All Pairs Shortest Distances (APSD)~\cite{sealfon2016shortest, ghazi2022differentially, fan2022breaking}. Likewise, the problem of  Differentially Private All Sets Range Query (ASRQ)~\cite{deng2023differentially} now has a tight error bound (up to logarithmic factors). We present the DP-APSD problem formally and show the proof of the new lower bound corresponding to \Cref{thm:intro-dp-apsd-lb}. Details on the DP-ASRQ problem are deferred to \Cref{apdx:asrq}.

\subsection{All Pairs Shortest Distances}
Given a weighted undirected graph $G = (V,E, w)$ of size $n$, the private mechanism is supposed to output an $n$ by $n$ matrix $D'=\mathcal{M}(G)$ of approximate all pairs shortest paths distances in $G$, and the privacy guarantee is imposed on two sets of edge weights that are considered `neighboring', i.e., with $\ell_1$ difference  at most $1$. Our goal is to minimize the maximum additive error of any entry in the APSD matrix, i.e., the $\ell_{\infty}$ distance of $D'-D$ where $D$ is the true APSD matrix. This line of work was initiated by \cite{sealfon2016shortest}, where an algorithm was proposed with $O(n)$ additive error. Recently, concurrent works \cite{ghazi2022differentially, fan2022breaking} breaks the linear barrier by presenting an upper bound of $\Oish(n^{1/2})$. Meanwhile, the only known lower bound is $\Omega(n^{1/6})$, due  to ~\cite{ghazi2022differentially}, using a hereditary discrepancy lower bound based on the point-line system of ~\cite{Chazelle2000-ld}. With our improved hereditary discrepancy lower bound, we are able to show an $\Omega(n^{1/4}/\sqrt{\log n})$ lower bound on the additive error of the DP-APSD problem. 

\begin{corollary}
\label{cor:dp-apsd-lb}
 Given an $n$-node undirected graph, for any $\beta \in (0,1)$ and $\varepsilon, \delta>0$, no $(\varepsilon, \delta)$-DP algorithm for APSD has additive error of $o(n^{1/4}/\sqrt{\log n})$ with probability $1-\beta$.
\end{corollary}

The connection between the APSD problem and the shortest paths hereditary discrepancy lower bound was shown in~\cite{ghazi2022differentially}, which implies that simply plugging in the new exponent gives the result above. For the sake of completeness, we give the necessary definition to formally define the DP-APSD problem, and show the main arguments towards proving \Cref{cor:dp-apsd-lb}.

\begin{definition}[Neighboring weights~\cite{sealfon2016shortest}]
\label{def:neighboring-weights}
    For a graph $G = (V,E)$, let $w, w':E\rightarrow \mathbb{R}^{\geq 0}$ be two weight functions that map any $e \in E$ to a non-negative real number, we say $w, w'$ are neighboring, denoted as $w \sim w'$ if $
        \sum_{e\in E}|w(e)-w'(e)| \leq 1.$
\end{definition}

\begin{definition}[Differentially Private APSD~\cite{sealfon2016shortest}]
    Let $w, w':E\rightarrow R^{\geq 0}$ be weight functions, and $\mathcal{A}$ be an algorithm taking a graph $G = (V,E)$ and $w$ as input. The algorithm $A$ is $(\varepsilon, \delta)$-differentially private on $G$ if for any neighboring weights $w \sim w'$ (See \Cref{def:neighboring-weights}) and all sets of possible output $\mathcal{C}$, we have:        $\Pr[\mathcal{A}(G, w)\in \mathcal{C}] \leq e^{\varepsilon}\cdot \Pr[\mathcal{A}(G, w') \in \mathcal{C} ]+\delta.$
    
    We say the private mechanism $\mathcal{A}$ is $\alpha$-accurate if the $\ell_\infty$ norm of $|\mathcal{A}(G,w)-f(G,w)|$ is at most $\alpha$, where $f$ indicates the function returning the ground truth shortest distances.
\end{definition}

\begin{proof}
[Proof of \Cref{cor:dp-apsd-lb}]
    First, suppose $\mathbf{A} \in \mathbb{R}^{{n \choose 2} \times n}$ is the shortest path vertex incidence matrix on the graph $G$. Previous work~\cite{ghazi2022differentially} has shown that the linear query problem on $\mathbf{A}$ can be reduced to the DP-APSD problem, formally stated as follows.

\begin{lemma}[Lemma 4.1 in~\cite{ghazi2022differentially}]
\label{lem:apsd-linearquery}
    Let $(V, \Pi)$ be a shortest path system with incidence matrix $\mathbf{A}$, if there exists an $(\varepsilon, \delta)$ DP algorithm that is $\alpha$-accurate for the APSD problem with probability $1-\beta$ on a graph of size $2|V|$, then there exists an $(\varepsilon, \delta)$ DP algorithm that is $\alpha$-accurate for the $\mathbf{A}$-linear query problem with probability $1-\beta$.
\end{lemma}

Now all we need to show is that the $\mathbf{A}$-linear query problem has a lower bound of $\Omega(n^{1/4}/\sqrt{\log(n)})$. We note the following result by~\cite{muthukrishnan2012optimal}.

\begin{lemma}
\label{lem:mn-disc-lb}
    For any $\beta \in (0,1)$, there exists $\varepsilon, \delta$ such that for any $\mathbf{A}$, no $(\varepsilon,\delta)$-DP algorithm is $\herdisc(\mathbf{A})/2$-accurate for the $\mathbf{A}$-query problem with probability $1-\beta$.
\end{lemma}

Combining \Cref{lem:apsd-linearquery}  and \ref{lem:mn-disc-lb}, we find that additive error needed for the DP-APSD problem is at least the hereditary discrepancy of its vertex incidence matrix, implying  \Cref{cor:dp-apsd-lb}. 
The lower bound for ASRQ also follows using the same argument (see \Cref{apdx:asrq}).
\end{proof}

\section{Conclusion and Open Problems}
This paper reported new  bounds on the hereditary discrepancy of set systems of unique shortest paths in graphs. We leave several open questions:
\begin{enumerate}
	\item An open problem  is to improve our edge discrepancy upper bound in directed graphs. Standard techniques in discrepancy theory imply an upper bound of $\min\{O(m^{1/4}), \widetilde{O}(D^{1/2})\}$ for this problem, leaving a gap with our $\Omega(n^{1/4}/\sqrt{\log n})$ lower bound when $m = \omega(n)$. Unfortunately, we were not able to extend our low-discrepancy edge and vertex coloring arguments for undirected graphs to the general directed setting, due to the pathological example in \Cref{fig:dir-coloring}. 
	\item Using our discrepancy lower bound, we gave an improved lower bound of $\widetilde \Omega(n^{1/4})$ on answering the all-pairs shortest distance problem under the constraints of $(1,1/n)$-differential privacy. In contrast, the best known upper bound remains $\widetilde O(n^{1/2})$ for the same problem. Closing this gap remains an interesting open question. 
\end{enumerate}

\begin{figure}[htp]
\centering
\includegraphics[width=0.35\linewidth]{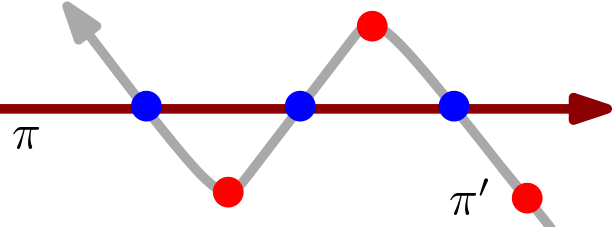}
\caption{
An example in directed graphs that demonstrates how coloring unique shortest paths with alternating colors can fail to imply low discrepancy.
}
\label{fig:dir-coloring}
\end{figure}

\bibliographystyle{plain}
\bibliography{references}

\appendix

\section{Discrepancy Bounds for Paths Without Consistency}
\label{apdx:discrepancy-path-inconsistent}
In a graph, when we consider simple paths without the consistency requirement, there is a strong lower bound on both vertex and edge discrepancy. In particular, we have the following theorem.

\begin{theorem} 
There is a planar graph $G=(V, E)$ such that the following is true for any coloring $f:  V \mapsto \{-1, 1\}$ of vertices $V$:
\begin{enumerate}
    \item\label{line:arbitrary-lb-exp} There is a family $\Pi$ of simple paths with $|\Pi|=O(\exp(n))$ and edge discrepancy of $\Omega(n)$.
    \item\label{line:arbitrary-lb-poly} There is a family $\Pi$ of simple paths with $|\Pi|=O(n)$ and edge discrepancy of $\Omega(\sqrt{n})$.
\end{enumerate}
The same claim holds true for vertex discrepancy as well. 
\end{theorem}
\begin{proof}
The first claim follows from Proposition 1.6 in~\cite{balogh2020discrepancies}, which says that the edge discrepancy of paths on a $k \times \ell$ grid graph is at least $\Omega(k\ell)$. To prove vertex discrepancy, we make two additional remarks about this construction. First, for our purpose, it is sufficient to consider only an $n \times 2$ grid graph $G$. The set of paths in the construction of~\cite{balogh2020discrepancies} consists of all paths that start from the top left corner and bottom left corner, going to the right, and possibly taking a subset of the vertical edges in the grid graph. The number of paths is $O(2^n)$ -- any subset of the $n$ vertical edges generates two paths. Second, for vertex discrepancy, we define a companion graph $G'$. Specifically, for each grid edge $e$ in $G$, we place a vertex $v$ of $G'$ on $e$. We connect two vertices in $G'$ if and only if the corresponding edges in $G$ share a common vertex. The graph $G'$ is still planar. In addition, a path $P$ in $G$ maps to a corresponding path $P'$ in $G'$ where the vertices on $P'$ follow the same order as the corresponding edges on $P$. See \Cref{fig:grid} for an example. Therefore, the edge discrepancy in $G$ and the vertex discrepancy of $G'$ are the same.  
\begin{figure}[htp]
\centering
\includegraphics[width=0.65\linewidth]{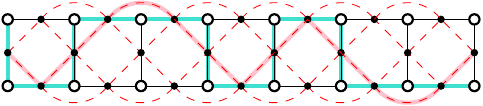}
\caption{An $n\times 2$ grid graph (with vertices shown in hollow and edges in black) $G$ with one path (in blue) starting from the top left corner go the right. The solid vertices and edges in dashed red define the companion graph $G'$. The corresponding path in $G'$ is shown in pink.}
\label{fig:grid}
\end{figure}

For the second claim, we take an $n\times n$ Hadamard matrix $H$ with $n$ as power of $2$. The elements in $H$ are $+1$ or $-1$. All the rows are pairwise orthogonal. For example, the Hadamard matrix with $n=8$ is
$$
H_8=\begin{pmatrix}
1 & 1 & 1 & 1 & 1 & 1 & 1 & 1 \\
1 & -1 & 1 & -1 & 1 & -1 & 1 & -1 \\
1 & 1 & -1 & -1 & 1 & 1 & -1 & -1 \\
1 & -1 & -1 & 1 & 1 & -1 & -1 & 1 \\
1 & 1 & 1 & 1 & -1 & -1 & -1 & -1 \\
1 & -1 & 1 & -1 & -1 & 1 & -1 & 1 \\
1 & 1 & -1 & -1 & -1 & -1 & 1 & 1 \\
1 & -1 & -1 & 1 & -1 & 1 & 1 & -1 \\
\end{pmatrix}
$$
It is known~\cite{Chazelle2000-rd} that the matrix $A=\{a_{ij}\}=\frac{1}{2}(H+J)$ with $J$ as an $n\times n$ matrix of all $1$ has discrepancy at least $\Omega(\sqrt{n})$. 

Now we try to embed the matrix $A$ by paths on a $2\times n$ grid graph $G$ with a grid of two rows and $n$ columns. Denote by $e_j$ as the $j$th vertical edge in $G$, $1\leq j\leq n$. For the $i$th row of $A$, we define set $X_i=\{e_j: a_{ij}=1\}$.
We then define two paths $P(X_i)$, $P'(X_i)$ on $G$. 
\begin{itemize}
    \item Path $P(X_i)$ starts from the top left corner of $G$ going to the right and the vertical edges visited by $P(X_i)$ are precisely $X_i$. 
    \item Path $P'(X_i)$ starts from the bottom left corner of $G$ going to the right and, similar to $P(X_i)$, the vertical edges visited by $P(X_i)$ are precisely $X_i$. 
\end{itemize}
Note that $P(X_i)$ and $P'(X_i)$ each contains edges $X_i$ (see \Cref{fig:signal-paths} for an illustration). Also $P(X_i)$ and $P'(X_i)$ do not share any horizontal edges, and collectively cover all horizontal edges in $G$.
In addition, we define two paths $P$ and $P'$ with $P$ starting from the top left corner and visiting all the top horizontal edges and $P'$ starting from the bottom left corner and visiting all bottom horizontal edges. We have $2n+2$ paths in total -- each row $i$ of the Hadamard matrix contributes $2$ paths, and we additionally use $P$ and $P'$.

\begin{figure}[htp]
\centering
\includegraphics[width=0.55\linewidth]{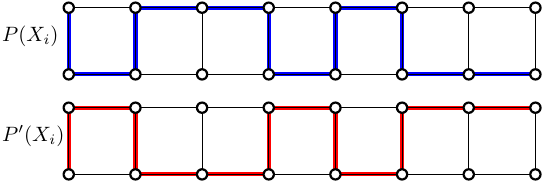}
\caption{A $n\times 2$ grid graph $G$ with two paths $P(X_i)$ and $P'(X_i)$. The $i$th row of matrix $A$, i.e., $A_{i} = (\frac{1}{2}(H+J))_{i}$, corresponds to $X_i = (1, 1, 0, 1, 1, 1, 0, 0)$.
}
\label{fig:signal-paths}
\end{figure}

For any $\{\pm 1\}$ coloring $f$ on edges in the grid graph $G$, define $D$ as the maximum absolute value of the sum of the colors of edges of each path, among all the $2n+2$ paths. We now argue that $D=\Omega(\sqrt{n})$.

First we define an $n$-dimensional vector $\chi \in \{+1, -1\}^n$ with $\chi_j=f(e_j)$, i.e., the color of the $j$th vertical edge in $G$. 
Since the discrepancy of matrix $A$ is $\Omega(\sqrt{n})$, there must be one row vector $a_i$ in $A$ such that $a_i \cdot \chi \geq c \sqrt{n}$ for some constant $c$. In other words, the sum of the colors of the edges in $X_i$ is $x=\sum_{e_j \in X_i} f(e_j)$ with $|x|\geq c\sqrt{n}$. Without loss of generality, we assume $x>0$, which in turn implies $x\geq c\sqrt{n}$.

Now we define \[z_1:=\sum_{e\in P(X_i)\setminus X_i} f(e) \quad \text{ and} \quad z_2:=\sum_{e\in P'(X_i)\setminus X_i} f(e).\]  
If $\max\{z_1, z_2\}\geq 0$, then we are done as the total sum of colors of either $P(X_i)$ or $P'(X_i)$ is at least $x\geq c\sqrt{n}$.
Otherwise, suppose we have $z_1<0$ and $z_2<0$. 
We now have from path $P(X_i)$, $D\geq x+z_1$, and from path $P'(X_i)$, $D\geq x+z_2$. Further, consider paths $P$ and $P'$; the sum of colors on all the horizontal edges is $z_1+z_2$. Since both $z_1$ and $z_2$ are negative, there should be \[2D\geq \left |\sum_{e\in P} f(e) \right |+ \left | \sum_{e\in P'} f(e) \right | \geq \left |\sum_{e\in P} f(e) + \sum_{e\in P'} f(e) \right | = -z_1-z_2.\] 
Summing up all three inequalities, we have $4D \geq 2x$. Thus $D\geq x/2\geq c\sqrt{n}/2$. This finishes the proof for edge discrepancy, and the vertex discrepancy bound can be obtained using the same trick as in the proof of claim 1.
\end{proof}

The above theorem provides lower bounds for the discrepancy. Since hereditary discrepancy is at least as high as discrepancy, the lower bounds also hold for hereditary discrepancy. 

{Grid graphs are bipartite graphs, so the vertex discrepancy of \emph{any} path in a grid graph is $O(1)$, by \Cref{lem:bipartite-graph}. In the proof above, our lower bound is for the \emph{edge} discrepancy of paths in a grid graph $G$, and the vertex discrepancy of the \emph{companion graph} $G'$. Notice that the companion graph $G'$ has odd cycles and is not bipartite.}

Nevertheless, the edge discrepancy of \emph{shortest paths} on grid graphs is actually a constant. Take \emph{shortest paths} on an unweighted grid graph (without even requiring the consistency property, and there could be exponentially many shortest paths between two vertices). 
All horizontal edges are given color $+1$ (or $-1$) if the left endpoint is at an even (or odd) $x$-coordinate and the right endpoint is at an odd (or even) $x$-coordinate. {We do the same for vertical edges, assigning colors $+1$ and $-1$ alternately according to the bottom vertex's odd/even $y$-coordinate. Any shortest path has a `staircase' shape and a total sum of colors of $O(1)$.} 

In contrast, the paths used in the above theorem are $2$-approximate shortest paths. 
{If we replace each horizontal edge by a chain of $\lceil 1/\eps \rceil$ vertices, we can make these paths to be $(1+\eps)$-approximate shortest paths for any $\eps>0$. Therefore, when we relax from shortest paths to approximate shortest paths, the discrepancy bounds increase substantially.} 

\section{Relation Between Discrepancy and Hereditary Discrepancy in General Graphs}
\label{apdx:disc-equal-herdisc}

In Theorem \ref{thm:undi-hdisc-lb}, we showed a construction of a path weighted graph $G$ whose system of unique shortest paths $\Pi$ satisfies $\herdisc_v(\Pi) \ge \Omega(n^{1/4}/\sqrt{\log(n)})$.
Here, we will observe that this result extends to discrepancy:
\begin{theorem} \label{thm:discvlb}
There are examples of $n$-node undirected weighted graphs $G$ with a unique shortest path between each pair of nodes in which this system of shortest paths $\Pi$ has
\[\discv(\Pi) \ge \Omega\left({n^{1/4}\over\sqrt{\log(n)}}\right).\]
\end{theorem}
\begin{proof}
Let $G$ be the graph from Theorem \ref{thm:undi-hdisc-lb}, and let $\Pi$ be its system of unique shortest paths.
By definition of hereditary discrepancy, there exists an induced path subsystem $\Pi' \subseteq \Pi$ with
\[\disc_v(\Pi') \ge \Omega\left({n^{1/4}\over\sqrt{\log(n)}}\right).\]
Recall that, by induced subsystem, we mean that we may view the paths of $\Pi$ as abstract sequences of nodes, and then $\Pi'$ is obtained from $\Pi$ by deleting zero or more nodes and deleting all occurrences of those nodes from the middle of paths.
Thus, the paths in $\Pi'$ are not still paths in $\Pi$, but they are paths in a different graph $G'$ on $ n' \le n$ nodes.
It thus suffices to argue that all paths in $\Pi'$ are unique shortest paths.
Indeed, this is well known, and is shown e.g.\ in \cite{Bodwin19} (c.f.\ Lemma 2.4.4 and 2.4.11).
To sketch the proof: suppose that a node $v$ is deleted from the initial system $\Pi$.
Consider each path $\pi \in \Pi$ that contains $v$ as an internal node, i.e., it has the form
\[\pi = (\dots, u, v, x, \dots).\]
When $v$ is removed, the path now contains the nodes $u, x$ consecutively, and so we must add $(u, x)$ as a new edge to $G'$ so that $\pi$ is a path in $G'$.
We judiciously set the edge weight to be $w(u, x) := w(u, v) + w(v, x)$.
The weighted length of the path $\pi$ does not change, and yet the distances of $G'$ majorize those of $G$, which implies that $\pi$ is still a unique shortest path in $G'$.
Inducting this analysis over each deleted vertex leads to the desired claim.
\end{proof}

\section{Application in Matrix Analysis}
\label{apdx:app-factorization-norm}
Hereditary discrepancy is intrinsically related to {\em factorization norm}, which has found applications in many areas of computer science, including but not limited to quantum channel capacity, communication complexity, etc. For any complex matrix $A \in \mathbb C^{m \times n}$, its factorization norm, denoted by $\gamma_2(A)$, is defined as the following optimization problem:
\begin{align}
\gamma_2(A) = \min \left\{  \|L\|_{2 \to \infty} \| R \|_{1 \to 2} : A = LR\right\}.
\label{eq:gamma2norm}    
\end{align}

One can write \cref{eq:gamma2norm} in a form of a semi-definite program (see Lee et al.~\cite{lee2008direct}) and also show that the Slater point exists. In particular, the primal and dual program coincides. An interesting question in matrix analysis is to estimate the factorization norm of different class of matrices. In a series of work, various authors have computed tight bounds on the factorization norm of certain class of matrices:
\begin{itemize}
    \item If $A$ is a unitary matrix, then $\gamma_2(A)=1$.
    \item If $A \in \mathbb R^{n \times n}$ is a positive semi-definite matrix with entries $A_{ij}$, then 
    \[\gamma_2(A) = \max_{1 \leq i \leq n} A_{ii}.\]
    \item Mathias~\cite{mathias1993hadamard}: $A$ satisfies the property that $\sqrt{A^\top A} \bullet \mathbb I = \sqrt{AA^\top} \bullet \mathbb I = {\tr(A) \over n} \mathbb I$, then 
    \[\gamma_2(A) = {\mathsf{Tr}(A^\top A) \over n}.\]
    In particular, if $A \in \{0,1\}^n$ is a lower-triangular one matrix, then $
    \gamma_2(A) = \Theta(\log n)$~\cite{fichtenberger2023constant,henzinger2023almost, kwapien1970main}.
    \item If $A \in \mathbb R^{n \times n}$ is a lower-triangular Toeplitz matrix with entries decreasing either polynomially or exponentially, then 
    $\gamma_2(A) = \Theta(1)$~\cite{henzinger2023unifying}.
\end{itemize}

Our tight bound on hereditary discrepancy for consistent path on graphs allows us to give tight bound on the factorization norm for the corresponding incidence matrix. In particular, we use the following result: 
\begin{lemma}
[Matou{\v s}ek et al.~\cite{matouvsek2020factorization}]
\label{lem:herdisc_factorization}
    For any real $m \times n$ matrix $A \in \mathbb R^{m \times n}$, there exists absolute constants $0<c<C$ such that
\[c\frac{\gamma_2(A)}{\sqrt{\log(m)}} \leq \herdisc(A) \leq C\gamma_2(A)\log(m).
\]
\end{lemma}

Combining \Cref{lem:herdisc_factorization} with our results, we have the following corollary:
\begin{corollary}
    Let $A_G$ be the incident matrix for unique shortest path system on an $n$ vertices graph $G$. Then if $G$ is bipartite, planar, or any general graph, then $\gamma_2(A_G) = \widetilde \Theta(n^{1/4})$.
\end{corollary}

\section{Differentially Private All Sets Range Queries} \label{apdx:asrq}
Here we introduce the DP-ASRQ problem and show its connection to the DP-APSD problem.

Given an undirected graph $G$, the problem of All Sets Range Queries (ASRQ) considers each edge associated with a certain attribute, and a range is the set of edges along a shortest path. 
A \emph{counting} range query returns the summation of the attributes along all edges in a range. 

At a schematic level, the ASRQ problem with counting queries is very similar to the APSD problem: the graph topology is public, and edge attributes are considered private. Only the shortest path structure is dictated by the edge weights in the APSD problem; however, it is irrelevant to the edge attributes in the ASRQ problem. This subtle difference has consequential caveat in the algorithm design: the graph topology can be used, for example, to construct an exact hopset first then apply perturbations to the edge attributes in the ASRQ problem; nevertheless, approach of this kind violates protecting the sensitive information in the APSD problem, since adversarial inference can be made on edge weights when the graph topology information is used. This observation also notes that the APSD problem is strictly harder than the ASRQ problem: recall that the best additive error upper bound of the APSD problem is still $\Oish(n^{1/2})$, while the other has $\Oish(n^{1/4})$~\cite{deng2023differentially}. Therefore, plugging in our new hereditary lower bound essentially closes the gap for the ASRQ problem.

\begin{corollary}[Formal version of \Cref{thm:intro-dp-apsd-lb}]
\label{cor:dp-asrq-lb}
 Given an $n$-node undirected graph, for any $\beta \in (0,1)$ and any $\varepsilon, \delta>0$, no $(\varepsilon, \delta)$-DP algorithm for ASRQ has additive error of $o(n^{1/4})$ with probability $1-\beta$.
\end{corollary}

\begin{definition}[Differentially Private Range Queries]
\label{def:dp-asrq}
Let $(\R=(X, \S), f)$ be a system of range queries and $w, w': X \rightarrow \mathbb{R}^{\geq 0} $ be neighboring attribute functions.  Furthermore, let $\mathcal{A}$ be an algorithm that takes $(\R, f , w)$ as input. Then $\mathcal{A}$ is $(\varepsilon, \delta)$-differentially private on $G$ if, for all pairs of neighboring attribute functions $w, w'$ and all sets of possible outputs $\mathcal{C}$, we have 
$\Pr[\mathcal{A}(\R, f , w)\in \mathcal{C}] \leq e^{\varepsilon}\cdot \Pr[\mathcal{A}(\R, f , w') \in \mathcal{C} ]+\delta.$
If $\delta=0$, we say $\mathcal{A}$ is $\varepsilon$-differentially private on $G $. 
\end{definition}

To complete this section, we give the formal definition of the ASRQ problem above.  The definition of neighboring attributes follows \Cref{def:neighboring-weights}. The lower bound proof of \Cref{cor:dp-asrq-lb} simply imitates the APSD problem, because the reduction from the linear query problem still holds despite the difference between the two problems. The proof is omitted to avoid redundancy.

\end{document}